\pdfoutput=1
\documentclass[pdftex,prl,aps,twocolumn,twoside,nofootinbib,showkeys,showpacs,10pt]{revtex4-2}

\usepackage[T1]{fontenc}
\fontencoding{T1}  
\usepackage[utf8]{inputenc}

\usepackage[normalem]{ulem}

\usepackage{todonotes}

\newcommand{\cT}{\mathcal{T}}
\newcommand{\cD}{\mathcal{D}}
\newcommand{\cB}{\mathcal{B}}
\newcommand{\cL}{\mathcal{L}}

\newcommand{\RR}{\mathbb{R}}
\newcommand{\CC}{\mathbb{C}}
\newcommand{\NN}{\mathbb{N}}
\newcommand{\cN}{\mathcal{N}}
\newcommand{\dom}{\operatorname{dom}}
\newcommand{\D}{\mathfrak{D}}

\newcommand{\tr}{\operatorname{tr}}

\newcommand{\cH}{\mathcal{H}}

\newcommand{\cF}{\mathcal{F}}

\renewcommand{\phi}{\varphi}
\renewcommand{\epsilon}{\varepsilon}

\usepackage{enumitem}
\usepackage{tikz}
\usepackage{amsfonts}
\usepackage{amsmath,amssymb,amsthm}
\usepackage{graphicx}
\usepackage{fancyhdr}
\usepackage[breakable]{tcolorbox}
\usepackage{bbm}
\usepackage{url}
\usepackage{mathtools}
\usepackage{braket}
\usepackage{upgreek }
\usepackage{dsfont}
\usepackage{collectbox}
\usepackage{braket}

\usepackage[plain]{fancyref} 
\usepackage[ruled]{algorithm2e}

\usepackage{algpseudocode}

\usepackage[colorlinks = true, citecolor = red, urlcolor=blue]{hyperref}

\usepackage{cleveref}

\newtheorem{thm}{Theorem}
\newtheorem*{thm*}{Theorem}
\makeatletter
\newcommand{\setthmtag}[1]{
  \let\oldthethm\thethm
  \newcommand{\thethm}{#1}
  \g@addto@macro\endthm{
    \addtocounter{thm}{-1}
    \global\let\thethm\oldthethm}
  }
\makeatother

\newtheorem{prop}[thm]{Proposition}
\newtheorem*{prop*}{Proposition}
\newtheorem{lemma}[thm]{Lemma}
\newtheorem*{lemma*}{Lemma}
\newtheorem{cor}[thm]{Corollary}
\newtheorem*{cor*}{Corollary}

\newtheorem*{cj*}{Conjecture}
\newtheorem{definition}[thm]{Definition}
\newtheorem*{Def*}{Definition}

\usepackage{titlesec}

\theoremstyle{definition}

\newtheorem*{rem*}{Remark}

\def\beq{\begin{equation}}
\def\eeq{\end{equation}}
\def\bq{\begin{quote}}
\def\eq{\end{quote}}
\def\ben{\begin{enumerate}}
\def\een{\end{enumerate}}
\def\bit{\begin{itemize}}
\def\eit{\end{itemize}}

\def\r|{\right|}

\newcommand{\cV}{\mathcal{V}}

\newcommand{\id}{\text{id}}

\newcommand{\cC}{\mathcal{C}}

\newcommand{\cR}{\mathcal{R}}

\newcommand{\cU}{\mathcal{U}}
\newcommand\be{\begin{equation}}
\newcommand\ee{\end{equation}}

\newcommand{\stkout}[1]{\ifmmode\text{\sout{\ensuremath{#1}}}\else\sout{#1}\fi}
\newif\ifverbose
\verbosetrue

\begin{document}
\title{\texorpdfstring{
		Limitations of local update recovery in stabilizer-$\mathsf{GKP}$ codes:\\
a quantum optimal transport approach}{Limitations of local update recovery in stabilizer-$\mathsf{GKP}$ codes: a quantum optimal transport approach}}

\author{\begingroup
\href{}{Robert K\"{o}nig
	\endgroup}
}
\affiliation{Munich Center for Quantum Science and Technology \& \\
Department of Mathematics, School of Computation, Information and Technology, Technical University of Munich, 85748 Garching, Germany}

\author{\begingroup
\href{https://orcid.org/0000-0001-7712-6582}{Cambyse Rouz\'{e}
\endgroup}
}
\email[Cambyse Rouz\'{e} ]{cambyse.rouze@tum.de}
 \affiliation{Inria, T\'{e}l\'{e}com Paris - LTCI, Institut Polytechnique de Paris, 91120 Palaiseau, France\\ Zentrum Mathematik, Technische Universit\"{a}t M\"{u}nchen, 85748 Garching, Germany}
 
\xdefinecolor{mygreen}      {RGB}{108,187,69}
\newcommand{\tm}[1]{{\color{mygreen}TM:~#1}}
\newcommand{\ca}[1]{{\color{blue}CR:~#1}}

\begin{abstract}
Local update recovery seeks to maintain quantum information by applying local correction maps alternating with and compensating for the action of noise. Motivated by recent constructions based on quantum LDPC codes in the finite-dimensional setting, we establish an analytic upper bound on the fault-tolerance threshold for concatenated $\mathsf{GKP}$-stabilizer codes with local update recovery. Our bound applies to noise channels that are tensor products of one-mode beamsplitters with arbitrary environment states, capturing, in particular, photon loss occurring independently in each mode.  It shows that for loss rates above a threshold given explicitly as a function of the locality of the recovery maps, encoded information is lost at an exponential rate.  This extends an early result by Razborov from discrete to continuous variable (CV) quantum systems. 

To prove our result, we study a metric on bosonic states akin to the Wasserstein distance between two CV density functions, which we call the bosonic Wasserstein distance. It can  be thought of as a CV extension of a quantum Wasserstein distance of order~$1$ recently introduced by De Palma et al. in the context of qudit systems, in the sense that it captures the notion of locality in a CV  setting.  We establish several basic properties, including a relation to the trace distance and diameter bounds for states with finite average photon number. We then study its contraction properties under quantum channels, including tensorization, locality and strict contraction under beamsplitter-type noise channels. Due to the simplicity of its formulation, and the established wide applicability of its finite-dimensional counterpart, we believe that the bosonic Wasserstein distance will become a versatile tool in the study of CV quantum systems.


\end{abstract}

\maketitle

\section{Introduction}

Universal quantum computers are envisioned to provide a powerful computational resource. Whether such devices will be available in the future depends on our ability to find experimentally achievable methods to protect the quantum information from the unavoidable presence of noise while simultaneously maintaining the ability to operate on it. The celebrated \textit{threshold theorem} proves that, at least in theory, such fault-tolerance constructions exist~\cite{Aharonov1997}, as long as the individual building blocks are sufficiently reliable: If e.g., the noise rate of individual qubits and physical gates is below a constant value called the \textit{fault-tolerance threshold}, then a size-$S$ quantum circuit can be emulated with error at most~$\epsilon$ by a noisy circuit of size polynomial in~$S$ and $\log 1/\epsilon$. Quantitative lower bounds of this threshold value (obtained by studying concrete fault-tolerance proposals for a given noise model) are of central importance as they provide clear experimental benchmarks. Conversely, upper bounds on the fault-tolerance threshold establish fundamental limitations on quantum computing and may guide the search for better fault-tolerance constructions.

\subsection{Fault-tolerant quantum memories}

{\it Fault-tolerant quantum memories:} To establish an upper bound on the fault-tolerance threshold, it suffices to consider the task of fault-tolerantly implementing the concatenation of~$T$ identity gates: In other words, this is the problem of realizing a noise-resilient memory that preserves quantum information for a time~$T$. To preserve $k$~logical qubits by using $n$~physical ones, a general error correction strategy relies on a quantum code~$\cC\subset (\mathbb{C}^2)^{\otimes n}$ encoding $k$~logical qubits into $n$~physical ones: Starting with an encoded version of the initial state, a recovery map~$\cR_t$ is applied to the~$n$~qubits at each time step~$t=1,\ldots,T$. The choice of~$\left(\cC,\cR=\{\cR_t\}_{t=1}^T\right)$ is intended to ensure that the information is preserved even if the execution of the recovery maps is interleaved with (sufficiently weak) noise. For concreteness, assume that each physical qubit undergoes depolarizing noise $\cN_p(\rho)=(1-p)\rho+p I/2$, with depolarizing parameter~$p\in (0,1)$.  Then we are interested in the (worst-case) fidelity with which the composed channel
\begin{align}
	\cR^{(T)}\coloneqq \bigcirc_{t=T}^{1} (\cR_t\circ\cN_p^{\otimes n})\ \label{eq:composedmap}
\end{align}
preserves states supported on~$\cC$, or more precisely, the trade-off relation between~$T$, this fidelity, and the numbers~$(n,k)$ of physical and logical qubits, respectively. Clearly, this trade-off  depends on the set of allowed recovery maps. We note that to make the notion of preserving states supported on~$\cC$ meaningful, one should generally let the final recovery map~$\cR_T$  be an arbitrary quantum channel even when~$\{\cR_j\}_{j=1}^{T-1}$ are restricted. This is similar to the definition of entanglement fidelity and ensures that we are quantifying the fidelity of the in-principle recoverable information after time~$T$.

As an example, if~$\cC$ is an $[n,k,d]$-stabilizer code with stabilizer generators~$\{S_j\}_{j=1}^{n-k}$, one could consider~$\cR_t=\cR$ for all $t$ with a recovery map~$\cR$ of the form
\begin{align}\label{eq:Recovery}
	\cR(\rho)=\sum_{s\in \{0,1\}^{n-k}}C(s)\Pi(s)\rho\Pi(s)C(s)^\dagger,
\end{align}
where $\Pi(s)=\prod_{j=1}^{n-k}\frac{1}{2}(I+(-1)^{s_j}S_j)$ projects onto the syndrome-$s$ subspace and $C(s)$ is a Pauli-correction classically computed from~$s$, which maps the state back to the code space. Then Eq.~\eqref{eq:composedmap} 
corresponds to a scenario where each correction and recovery operation is ideally realized. In the case of probabilistic Pauli noise, the logical overall error model after time $T$ is then obtained by taking the $T$-fold convolution of the effective logical noise associated with a single step and recovery combination~$\cR\circ\cN_p^{\otimes n}$.

\subsection{Local update recovery}
A more interesting case -- the one considered here -- is where locality is imposed on the operators~$\{\cR_t\}_{t=1}^{L-1}$.  Specifically, let us say that the recovery scheme~$(\cC,\cR=\{\cR_t\}_{t=1}^T)$ is $\ell$-local if for each~$t=1,\ldots,T-1$, the map~$\cR_t={\tiny{\prod_{j}\cR_{A^{(t)}_j}}}$ is a product of quantum channels~${\tiny{\cR_{A^{(t)}_j}}}$ each acting on pairwise disjoint subsets~$A^{(t)}_j$ of at most~$|A^{(t)}_j|\leq \ell$ qubits $A^{(t)}_j\subset\{1,\ldots,n\}$. Clearly, an~$O(1)$-local recovery scheme is attractive from an experimental point of view, but it should also be noted that this kind of recovery is highly restricted. For example, in the case of a stabilizer codes, it is not generally sufficient to have local (i.e., low-weight) stabilizer generators: even if syndrome extraction is achieved by measuring constant-size subsets of qubits, typical decoders pool and globally process corresponding measurement results to determine a suitable (global) correction. In contrast, an $O(1)$-local recovery scheme requires that syndrome extraction, computation and application of an associated correction is achieved by a collection of independent processes acting on disjoint subsets of qubits. We refer to such an $O(1)$-local scheme as a local update recovery.

Remarkably, despite the highly restrictive locality requirement, $O(1)$-local recovery schemes exhibiting a threshold property exist. Indeed, recent work focusing on the break-through discovery of quantum low density parity check (LDPC) codes~\cite{fawzi2018efficient} establishes a threshold property for such a scheme under a more general noise model than independent depolarizing noise on each qubit \cite{gu2023single}. In fact, the result of~\cite{fawzi2018efficient} applies even to the case where the implementation of each recovery map~$\cR_t$ is noisy. The construction is based on the fact that expander codes, i.e.,  codes whose Tanner graph is an expander graph, can be decoded by a local process, a fundamental insight first observed by Sipser and Spielman~\cite{sipser1996expander}. While the original decoder iteratively goes through the Tanner graph, extracting local syndrome information and flipping bits to decrease the syndrome weight, this process can be parallelized.
These ideas were introduced to the quantum context in Refs.~\cite{leverrier2015quantum}, where the so-called small-set-flip decoder was introduced and applied in a fault-tolerance construction. Remarkably, the new results of~\cite{fawzi2018efficient} show that for certain quantum LDPC codes, this locally defined error correction process works even if recovery maps are interspersed with noise maps as in~\eqref{eq:composedmap}. It should be noted that this analysis is different from prior works involving local updates such as Toom's rule~\cite{pastawski2011quantum,pastawski2012quantum} or certain cellular automata~\cite{breuckmann2016local,harrington2004analysis,Dennis2002,herold2015cellular}: The latter studies consider decoding/recovery by iterative application of local maps (possibly assisted by local memories), but  do not analyze the effect of additional errors within iterations.

\subsection{Bounding local update recovery for qubits}

While Ref.~\cite{fawzi2018efficient} gives a scheme for local update recovery for noise rates~$p<p_0$ below some threshold~$p_0$, here we ask for upper bounds on the threshold. In fact, such an upper bound is known \cite{Kempe10,razborov2003upper}: For any $\ell$-local recovery scheme~$(\cC,\cR=\{\cR_t\}_{t=1}^T)$, if the depolarizing probability is above $1-1/\ell$, then at any depth $T\ge C\log(n)$ for a sufficiently large constant $C>0$ and for any two input states $\rho,\sigma$:
\begin{align}\label{eq:rhosigmaupperbound}
	\big\|\cR^{(T)}(\rho)-\cR^{(T)}(\sigma)\big\|_1=\mathcal{O}(1/n)\,.
\end{align}
In words, the circuit loses distinguishability (and hence encoded information) at an exponential rate in~$T$ for any depolarizing probability
\begin{align}\label{depol}
	p>1-\frac{1}{\ell}=:p_0(\ell)\  .
\end{align}
We refer to Section~\ref{qubitproof} below for more details on the bound~\eqref{eq:rhosigmaupperbound}.

\subsection{Continuous variable error correction}

While the previous discussion concerned discrete variable (DV) systems, many systems of practical interest involve continuous variables (CV),
such as harmonic oscillators. The last twenty years have seen the development of various CV codes, e.g. encodings into Fock states \cite{chuang1997bosonic,Knill2001,Michael2016}, coherent states \cite{Cochrane99,Leghtas13}, or position and momentum eigenstates \cite{Lloyd98,Braunstein1998,Gottesman01}. Among these, the so-called Gottesman-Kitaev-Preskill ($\mathsf{GKP}$) codes \cite{Gottesman01} hold a prominent position.

By design, a $\mathsf{GKP}$ code corrects small random shifts in position and momentum. In contrast, the pure loss channel is the most common incoherent error process in optical and microwave cavities \cite{victor16}. Despite this apparent discrepancy, recent numerical studies showed that $\mathsf{GKP}$ codes significantly outperform all other codes for most values of the loss rate \cite{Albert18}. In  \cite{Albert18}, these  numerical results were also backed by an analytical upper bound on the channel infidelity between the noise composed with the $\mathsf{GKP}$ recovery channel and the ideal identity channel. While a CV quantum error-correcting code with a single bosonic mode and a single ancilla qubit like the $\mathsf{GKP}$ code can suppress relevant errors such as photon losses in a hardware-efficient manner, it should also be noted that logical error rates cannot be suppressed to an arbitrarily small value with this minimal architecture. To further suppress the residual errors, these bosonic codes should, for example, be concatenated with a qubit error-correcting code \cite{bravyi1998quantum,Dennis2002}. Recent proposals include concatenation with a repetition code \cite{Fukui17}, the surface code \cite{wang2019quantum,Fukui2018,Vuillot2019} or cluster states in the setting of measurement-based quantum computing \cite{Fukui2018,Menicucci14,fukui2019high}. In those works, the code capacity thresholds were obtained by assuming that only qubits can fail, i.e., gates, state preparations, and measurements are assumed perfect. A more realistic approach was recently taken for the surface code in \cite{Noh2020}. However, these works concentrate on upper bounds on the threshold noise.

The goal of the present paper is to provide an analytic lower bound on the accuracy threshold for $\mathsf{GKP}$-stabiliser concatenated codes with local update recovery, realistic input state preparation and measurements and in presence of photon-loss, akin to the threshold \eqref{depol}. 

\section{Bosonic local update recovery}

\subsection{The setup}

Consider a local recovery scheme~$(\cC,\cR=\{\cR_t\}_{t=1}^T)$ using $n$~physical qubits, based on an $[n,k]$-stabilizer code~$\cC$ and recovery maps~$\cR_t$ given by Clifford circuits as that of~\cite{Gottesman01}.
Using code concatenation, such a scheme can  immediately be lifted to a scheme operating on $n$~bosonic modes. Here we consider the ``standard''  Gottesman-Kitaev-Preskill ($\mathsf{GKP}$) code encoding a single qubit into a single harmonic oscillator. Concatenating the stabilizer~$\cC$ with the $\mathsf{GKP}$ code then provides a stabilizer-$\mathsf{GKP}$-code~$\cC_{\mathsf{osc}}$ encoding $k$~logical qubits into $n$~oscillators. Examples include the surface- or toric-$\mathsf{GKP}$ codes~\cite{noh2020fault,vuillot2019quantum}. A natural recovery (error correction) map for this concatenated code proceeds by applying a bosonic recovery map~$\cR_{\mathsf{GKP}}$ to each mode, and subsequently applying~$\overline{\cR}$ to the collection of $n$~modes. Here~$\overline{\cR}$ is a $\mathsf{GKP}$-encoded version of the original recovery map~$\cR$ for the stabilizer code~$\cC$. We write~$\cR^{\mathsf{osc}}\coloneqq \overline{\cR}\circ\cR_{\mathsf{GKP}}^{\otimes n}$ for this composed map. Applying this transformation 
to each recovery map~$\cR_t$ yields an $\ell$-local recovery scheme~$(\cC_{\mathsf{osc}},\cR_{\mathsf{osc}}=\{\cR^{\mathsf{osc}}_t\}_{t=1}^T)$, whose operations act on disjoint subsets of $\ell$ bosonic modes. To make this construction physical and concrete, we use approximate $\mathsf{GKP}$-states~$\ket{\tilde{0}}$ and~$\ket{\tilde{1}}$ with finite average energy, and a Steane-type error correction circuit for~$\cR_{\mathsf{GKP}}$ relying on unsharp (finite-variance) position- and momentum measurements (see Figure \ref{fig:steane} in the supplementary material).

Here we study to which degree bosonic local update recovery schemes constructed in this way can withstand photon loss. That is, we consider loss of a photon modeled by a one-mode beamsplitter of transmissivity $\lambda\in (0,1)$ with arbitrary environment state $\sigma_E$:
\begin{align}\label{beamsplitter}
	\mathcal{N}_\lambda(\rho)=\tr_{E}\big[U_\lambda(\rho\otimes \sigma_E)U_\lambda^\dagger\big]\,,
\end{align}
with the beamsplitter unitary $U_\lambda\equiv \operatorname{exp}\big((a^\dagger b-b^\dagger a)\operatorname{arccos}(\sqrt{\lambda})\big)$. 
Although this is not required for our results to hold, we will fix the state $\sigma_E$ to be the same each time the noise channel $\cN_\lambda$ is applied, for sake of simplicity. We are interested in the map
\begin{align}
	\cR^{(T)}_{\mathsf{osc}}\coloneqq \bigcirc_{t=T}^{1} (\cR_t^{\mathsf{osc}}\circ\cN_\lambda^{\otimes n})\ ,
\end{align} where each mode undergoes photon loss independently of the other modes, and this noise process is alternated with  recovery maps obtained by concatenation.

\subsection{Bounding local update recovery for concatenated $\mathsf{GKP}$-stabilizer codes}

Our main result is an upper bound for the bosonic setup akin to Eq.~\eqref{eq:rhosigmaupperbound}. It states the following:
\begin{thm}[Limits to local update recovery]\label{mainresult} 
	Let $(\cC_{\mathsf{osc}},\cR_{\mathsf{osc}}=\{\cR^{\mathsf{osc}}_t\}_{t=1}^T)$ be a bosonic recovery scheme on~$n$ modes derived by code concatenation from a (stabilizer-type) recovery scheme~$(\cC,\cR=\{\cR_t\}_{t=1}^{T})$. Then there exist constants~$\kappa$ depending on the noise parameters $\lambda,\sigma_E$, and~$C$ depending only on
	~$\{\cR_t\}_{t=1}^{T}$ such that 
	\begin{align}
		\left\|\cR^{(T)}_{\mathsf{osc}}(\rho)-\cR^{(T)}_{\mathsf{osc}}(\sigma)\right\|_1 & \leq 
		4\kappa\,\big(C\,\sqrt{\lambda}\big)^T\sqrt{2n(n+E)}\,,
	\end{align}
	for any two $n$-mode quantum states $\rho,\sigma$ with $\tr[\rho N],\tr[\sigma N]\le E$, and any photon loss strength (transmissivity)~$\lambda\in (0,1)$.
	Furthermore, if~$(\cC,\cR=\{\cR_t\}_{t=1}^{T})$ is an~$\ell$-local recovery scheme, then~$C=O(\ell)$. 
\end{thm}
Theorem~\ref{mainresult} implies, in particular, that 
the distinguishability between the outputs $\cR^{(T)}_{\mathsf{osc}}(\rho)$ and $\cR^{(T)}_{\mathsf{osc}}(\sigma)$ decays exponentially in the number~$T$ of time-steps whenever
\begin{align}
	\lambda<\frac{1}{C^2}\,
\end{align}
thus establishing an upper bound on the fault-tolerance threshold of the construction against photon loss errors.

\section{Quantum optimal transport toolbox}

\subsection{Qubit threshold by quantum optimal transport}\label{qubitproof}

Before we introduce our main new tool in the proof of \Cref{mainresult}, let us provide a modern (albeit with slightly worse constants) treatment of the result of Razborov \cite{razborov2003upper} via the use of recently introduced quantum optimal transport metrics \cite{DePalma2021}. The Wasserstein-1 distance between any two $n$-qubit quantum states $\rho,\sigma$ is defined as 
\begin{align}\label{eqW1}
	W_1(\rho,\sigma)\coloneqq \sup_{\|X\|_L\le 1}\,\Big|\tr\big[X(\rho-\sigma)\big]\Big|\,.
\end{align}
The optimization in \eqref{eqW1} is over observables $X$ of Lipschitz constant 
\begin{align}
	\|X\|_L\coloneqq 2\max_{i\in[n]}\, \min_{X^{(i)}}\,\big\|X-X^{(i)}\otimes I_i\big\|_\infty
\end{align}
less than or equal to $1$, where the minimization above is over observables $X^{(i)}$ supported on $[n]\backslash \{i\}$. As opposed to the trace distance, where the optimization is over observables of operator norm $\|X\|_\infty\le 1$, the Wasserstein distance is able to capture local differences between the states $\rho$ and $\sigma$. In particular, whenever $\rho$ and $\sigma$ share the same marginals over a set of $k$ qubits, 
\begin{align}\label{eq1W1}
	\frac{1}{2}\|\rho-\sigma\|_1\le W_1(\rho,\sigma)\le \,\frac{3(n-k)}{4}\,\|\rho-\sigma\|_1\,.
\end{align}
This results in the following control on the Wasserstein distance at the outputs of any $\ell$-local recovery map $\mathcal{R}$ and its inputs: 
\begin{align}\label{eq2W1}
	W_1(\cR(\rho),\cR(\sigma))\le \frac{3\ell}{2}\,W_1(\rho,\sigma)\,.
\end{align}
Moreover, the Wasserstein distance contracts after each layer of independent depolarizing noise:
\begin{align}\label{eq3W1}
	W_1(\cN_p^{\otimes n}(\rho),\cN_p^{\otimes n}(\sigma))\le (1-p) W_1(\rho,\sigma)\,.
\end{align}
Using the bounds \eqref{eq1W1}, \eqref{eq2W1} and \eqref{eq3W1} iteratively, we have
\begin{align*}
	\big\|\cR^{(T)}(\rho-\sigma)\big\|_1&\le 2\,W_1(\cR^{(T)}(\rho),\cR^{(T)}(\sigma))\\
	&\le 2\, \Big(\frac{3\ell (1-p)}{2}\Big)^T\, W_1(\rho,\sigma)\\
	&\le  \frac{3n}{2}\, \Big(\frac{3\ell (1-p)}{2}\Big)^T\,\|\rho-\sigma\|_1\\
	&\le 3n\,\Big(\frac{3\ell (1-p)}{2}\Big)^T\,.
\end{align*}
Therefore, if $p> 1-\frac{2}{3\ell}$, the above bound implies \eqref{eq:rhosigmaupperbound}.

\subsection{A new bosonic Wasserstein distance}

Inspired by the argument laid out in \Cref{qubitproof}, we aim at introducing a distance on $n$-mode quantum states of a bosonic system $\Lambda$ that satisfies similar bounds as \eqref{eq1W1}, \eqref{eq2W1} and \eqref{eq3W1}, namely: (i) a lower bound in terms of the trace distance as well as a uniform upper bound for a large enough class of physically relevant states; (ii) an input-output control for states undergoing a layer of $\ell$-local $\mathsf{GKP}$-stabilizer error correction; and (iii) a strict contraction under the action of the noise channel $\cN_\lambda^{\otimes n}$. A natural candidate in the classical, continuous-variable framework is the distance dual to the Lipschitz constant $\|\nabla f\|\coloneqq \max_{j\in [n]}\|\partial_{x_j}f\|_\infty$ on the class of bounded, differentiable functions $f:\mathbb{R}^n\to\mathbb{R}$. By a standard replacement of partial derivatives $\partial_{x_j}$ by commutators $[P_j,(\cdot)]$, $[Q_j,(\cdot)]$ with the position and momentum operators $P_j,Q_j$, $j\in[n]$, we arrive at the following definition.
\begin{definition}
	The bosonic Lipschitz constant of an observable $X$ on $\Lambda$ is defined as
	\begin{equation}
		\|\nabla X\|  \coloneqq\,\sup\,|\langle \psi|[R_j,X]\otimes I_R|\varphi\rangle | \equiv \max_j\,\|\nabla_j X\| \,,\nonumber
	\end{equation}
	where the supremum is over all reference systems $R$, all quadrature operators $R_j\in\{P_j,Q_j\},\,j\in \Lambda$, and any two pure states $|\psi\rangle,|\varphi\rangle$ of the joint system $\Lambda R$ with finite total photon number on $\Lambda$. When this is finite, it coincides with $\max_{R_j}\|[R_j,X]\|_\infty$. We then define the \textit{bosonic Wasserstein distance} between two $n$-mode quantum states $\rho,\sigma$ as
	\begin{equation}\label{dualWasserstein}
		W_{\operatorname{B}}(\rho,\sigma)\coloneqq\sup_{\|\nabla X\|\le 1}\,\Big|\tr\big[X(\rho-\sigma)\big]\Big|\,.
	\end{equation}
	where the supremum is over all bounded observables $X$.
\end{definition}

In the supplementary material, we show that the distance $W_{\operatorname{B}}$ satisfies bounds analogous to \eqref{eq1W1}, \eqref{eq2W1} and \eqref{eq3W1} in the qubit setting. To show a bosonic version of \eqref{eq1W1}, it would be enough to have that, for any observable $X$, $\|\nabla X\|_\infty\le C\|X\|_\infty$ for some constant $C>0$. Indeed, this would imply that $\|\rho-\sigma\|_1\le C\,W_{\operatorname{B}}(\rho_1,\rho_2)$. In analogy with the classical, continuous variable setting, this condition should, however, not hold since one can easily construct bounded, non-differentiable functions. However, when convoluted with a smooth density $g$, any function $f$ acquires the regularity of $g$, so that $\|\nabla (f\ast g)\|_\infty\le C\|f\|_\infty$. Since the noise channel $\cN_\lambda^{\otimes n}$ can be thought of as a quantum convolution \cite{konig2014entropy}, one can expect that a similar smoothing occurs in the bosonic framework. This intuition can be rigorously formalized and leads to the following bound:
\begin{prop}[see \Cref{regularity}]
	For any two $n$-mode states $\rho,\sigma$ and $\lambda\in (0,1)$, we have 
	\begin{align}\label{eqT1W1}
		\|\cN^{\otimes n}_\lambda(\rho_1-\rho_2)\|_1\le \kappa \, W_{\operatorname{B}}(\rho_1,\rho_2)\,,
	\end{align}
	where $\kappa\coloneqq \sqrt{\frac{\lambda}{1-\lambda}}\,\max\left\{\big\|[Q,\sigma_{E}]\big\|_1,\big\|[P,\sigma_{E}]\big\|_1\right\}$.
\end{prop}

In turn, the bosonic Wasserstein distance can be controlled in terms of the average photon number of the states $\rho,\sigma$:

\begin{prop}[see \Cref{W1toT1}]
	For any two $n$-mode states $\rho,\sigma$ of average total photon number bounded by $ E\ge 0$, 
	\begin{align}
		W_{\operatorname{B}}(\rho,\sigma)\le 4 \sqrt{2\,n\,(n+E)}\,.\label{eq:diameterbound}
	\end{align}
\end{prop}
The proof of the bound \eqref{eq:diameterbound} relies on a semigroup interpolation result detailed in \Cref{diameterboundsec}. 
Next, we can use the canonical commutation relations in order to find the following bosonic analogue of \eqref{eq3W1}:
\begin{prop}[see \Cref{prop:contractionNlambda}] For any $\lambda\in [0,1]$ and any two $n$-mode state $\rho,\sigma$
	\begin{align}
		W_{\operatorname{B}}(\cN_\lambda^{\otimes n}(\rho),\cN_\lambda^{\otimes n}(\sigma)) \le \sqrt{\lambda}\,W_{\operatorname{B}}(\rho,\sigma)\,.\label{NlW11}
	\end{align}
\end{prop}

Finally, we would like to find bosonic analogues of \eqref{eq2W1}. First, we show that the $1$-mode GKP error correction scheme satisfies an input-output control bound akin to \eqref{eq2W1}. We refer the interested reader to \Cref{Gauss} for a short review on $\mathsf{GKP}$ error correction.

\begin{prop}[see \Cref{prop:ECW1}]
	The $1$-mode $\mathsf{GKP}$ quantum error correction scheme modeled by the channel $\cR_{\mathsf{GKP}}$ introduced in \Cref{Gauss} satisfies
	\begin{align}\label{GKPbound}
		W_{\operatorname{B}}(\cR_{\mathsf{GKP}}^{\otimes n}(\rho),\cR_{\mathsf{GKP}}^{\otimes n}(\sigma))\le 2 W_{\operatorname{B}}(\rho,\sigma)
	\end{align}
	for all $n$-mode states $\rho,\sigma$.
\end{prop}

At this stage, combining \eqref{GKPbound} together with \eqref{NlW11}, one can show that the Wasserstein distance decreases exponentially with the circuit depth whenever $\lambda\le 1/4$. This is perhaps not so surprising since the quantum capacity of a quantum loss channel is known to reach $0$ at $\lambda=1/2$ \cite{Wolf07}. However, it is interesting to note that the convergence in Wasserstein or trace distance that we establish here implies the stronger convergence of any (energy-constrained) capacity of the noisy circuit $\mathcal{C}$, and in particular of its energy constrained classical capacity \cite{Shirokov2018}. We also mention that an upper bound on the channel infidelity between the noise composed with the $\mathsf{GKP}$ recovery channel and the ideal identity channel
as a function of $\lambda$ was found in \cite[Equation (7.24)]{Albert18}.

Finally, we consider recovery maps obtained from code concatenation. Specifically, we consider an $[n,k,d]$-stabilizer code $\cC$ - where $n$ denotes the number of physical qubits, $k$ the number of logical qubits and $d$ the code distance- with stabilizer generators~$\{S_j\}_{j=1}^{n-k}$  and an associated recovery map~$\cR$ of the form given in \eqref{eq:Recovery}. We are interested in local recovery operations, and formalize this as follows: We assume that there is a  Clifford circuit~$U$ (composed of one- and two-qubits) on the system and additional~$r=n-k$ auxiliary qubits $A_1\dots A_{n-k}$ such that syndrome information can be extracted by applying~$U$ and subsequently measuring each individual qubit.
We capture locality as follows: first, we assume that the backward lightcone of every 
ancilla qubit $A_j$ under the action of $U$ has size at most $\ell_{\mathrm{meas}}$.
Furthermore, we assume that the correction operation can be computed by applying local functions to~$s$: There is a partition $[n]={\bigcup}_{j=1}^r \cF_j$ of $[n]$ into disjoint subsets~$\cF_1,\ldots,\cF_r$ and functions $C_j:\{0,1\}^{\cF_j}\rightarrow\{I,X,Y,Z\}^{\otimes |\cF'_j|}$ such that 
\begin{align}
	C(s)&=\prod_{j=1}^r C_j(s_{\cF_j})_{\cF'_j}\qquad\textrm{ for all }\qquad s\in \{0,1\}^{n-k}\ ,\label{eq:correctionlocalfunction}
\end{align}
where $s_{\cF_j}$ denotes the restriction of~$s$ to~$\cF_j$, i.e., the corresponding substring of syndrome bits. That is, for each~$j=1,\ldots,r$,  each operator~$C_j(s_{\cF_j})$ has support contained in a set~$\cF'_j$. For what follows, we further assume that the sets $\{\cF'_j\}_j$ are pairwise disjoint. Then, we suppose that there are constants~$\ell_{\mathrm{corr}},\ell'_{\mathrm{corr}}$ such that $|\cF_j| \leq \ell_{\mathrm{corr}}$, resp.~ $|\cF'_j|  \leq \ell'_{\mathrm{corr}}$ for all $j=1,\ldots,r$.

\begin{prop}[see \Cref{prop:upperboundlocalrecovery}]
	The recovery map $\overline{\cR}$ satisfies  
	\begin{align}\label{refGamma}
		W_{\operatorname{B}}(\overline{\cR}(\rho),\overline{\cR}(\sigma))\le  \Gamma \,W_{\operatorname{B}}(\rho,\sigma)
	\end{align}
	for all $n$-mode states $\rho,\sigma$, where 
	$$\Gamma \coloneqq \ell_{\operatorname{meas}}+ \frac{\ell_{\operatorname{meas}}\,2^{1+\ell_{\operatorname{meas}}(1+\ell_{\operatorname{corr}})}\ell_{\operatorname{meas}}\cdot \ell_{\operatorname{corr}}'}{\sqrt{\alpha_{\min}\pi}}\,,$$
	where $\alpha_{\operatorname{min}}$ denotes the minimal variance of the unsharp position and momentum measurements done in the scheme. 
\end{prop}

Finally, by combining \eqref{eq:diameterbound}, \eqref{refGamma}, \eqref{GKPbound}, \eqref{NlW11} and \eqref{eqT1W1}, we have proved that, for any two states $\rho,\sigma$ of average total photon number bounded by $E\ge 0$,
\begin{align*}
	\|\mathcal{R}_{\mathsf{osc}}^{(T)}(\rho-\sigma)\|_1\le 4\kappa\,\big(2\Gamma\,\sqrt{\lambda}\big)^T\sqrt{2n(n+E)}\,,
\end{align*}
which is what we claimed in \Cref{mainresult} for $C=2\Gamma$.

\section*{Acknowledgments}
R.K.  gratefully acknowledges support by the European Research Council under grant agreement no.~101001976 (project EQUIPTNT).

C.R. acknowledges the support of the Munich Center for Quantum Science and Technology, as well as that of the Humboldt Foundation.

\bibliographystyle{ieeetr}
\bibliography{biblio}

\onecolumngrid

\newpage 

\appendix
\section*{Supplementary material}\label{sec:proofnoisycircuit}

\section{Notations and definitions}

\subsection{Operators and norms}

Given a separable Hilbert space $\cH$, we denote by $\cB(\cH)$ the space of bounded linear operators on $\cH$, and by $\cT_p(\cH)$, the \textit{Schatten $p$-class}, which is the Banach subspace of $\cB(\cH)$ formed by all bounded linear operators whose Schatten $p$-norm, defined as $\|X\|_{p}=\left(\tr|X|^p\right)^{1/p}$,  is finite. Henceforth, we refer to $\cT_1(\cH)$ as the set of \textit{trace class} operators. The set of quantum states (or density matrices), that is positive semi-definite operators $\rho \in \cT_1(\cH)$ of unit trace, is denoted by $\cD(\cH)$. The Schatten $1$-norm, $\|\cdot\|_1$, is the {trace norm}, and the corresponding induced distance (e.g.\ between quantum states) is the {trace distance}. Note that the Schatten $2$-norm, $\|\cdot\|_2$, coincides with the \textit{Hilbert--Schmidt norm}.

For a pair of positive semi-definite operators, $X,Y$ with domains $\dom(X),\dom(Y) \subseteq \cH$, $X\geq Y$ if and only if $\dom\left(X^{1/2}\right)\subseteq \dom\left(Y^{1/2}\right)$ and $\left\|X^{1/2}\ket{\psi}\right\|^2\geq \left\|Y^{1/2}\ket{\psi}\right\|^2$ for all $\ket{\psi}\in \dom\left(X^{1/2}\right)$. 
If $\rho$ is a quantum state with spectral decomposition $\rho=\sum_i p_i |{\phi_i}\rangle\langle \phi_i|$, and $X$ is a positive semi-definite operator, the \textit{expected value} of $X$ on $\rho$ is defined as
\begin{equation}\tr[\rho X]\coloneqq \sum_{i:\, p_i>0} p_i \left\|X^{1/2}|{\phi_i}\rangle \right\|^2 \in \RR_+\cup \{+\infty\}\, ;
	\label{expected positive}
\end{equation}
here we use the convention that $\tr[\rho X]=+\infty$ if the above series diverges or if there exists an index $i$ for which $p_i>0$ and $\ket{\phi_i}\notin \dom\left(X^{1/2}\right)$. This definition can be extended to a generic densely defined self-adjoint operator $X$ on $\cH$, by considering its decomposition $X=X_+-X_-$ into positive and negative parts, with $X_\pm$ being positive semi-definite operators with mutually orthogonal supports. The operator $X$ is said to have a \textit{finite expected value on $\rho$} if $(i)$ $\ket{\phi_i}\in \dom\big(X_+^{1/2}\big)\cap \dom\big(X_-^{1/2}\big)$ for all $i$ for which $p_i>0$, and $(ii)$ the two series $\sum_i p_i \big\|X_\pm^{1/2} \ket{\phi_i}\big\|^2$ both converge. In this case, the following quantity is called the \textit{expected value} of $X$ on $\rho$:
\begin{equation}
	\tr[\rho X]\coloneqq \sum_{i:\, p_i>0} p_i \left\|X_+^{1/2} \ket{\phi_i}\right\|^2 + \sum_{i:\, p_i<0} p_i \left\|X_-^{1/2} \ket{\phi_i}\right\|^2
	\label{expected}
\end{equation}
Obviously, for a pair of operators $X,Y$ satisfying $X\geq Y$, we have that $\tr[\rho X]\geq \tr[\rho Y]$.

Here we adopt standard notations from quantum information theory: given a multipartite quantum system with associated Hilbert space $\cH=\cH_A\otimes \cH_B$, and a state $\rho\in \cD(\cH)$, we denote by $\rho_A\coloneqq\tr_{\cH_B}(\rho)$ the marginal state on system $A$. Similarly, an observable, i.e. a self-adjoint operator $X$ on $\cH_A$ is identified with the observable $X_A\otimes I_B$ on the system $AB$, that is over the joint Hilbert space $\cH_A\otimes \cH_B$, where $I_B$ denotes the identity operator on $\cH_B$. A quantum channel with input system $A$ and output system $B$ is any completely positive, trace-preserving (CPTP) linear map $\cN:\cT_1(\cH_A)\to\cT_1(\cH_B)$, where $\cH_A, \cH_B$ are the Hilbert spaces corresponding to $A,B$, respectively. We denote identity superoperator over a system $A$ by $\id_A$.  

\subsection{Phase-space formalism}

In this paper, given $m\in\NN$, we are concerned with the Hilbert space $\cH_m\coloneqq L^2(\RR^m)$ of a so-called $m$-mode oscillator, which is the space of square-integrable functions on $\RR^m$. We often denote by $\Lambda=\{1,\dots,m\}$ the (arbitrarily ordered) set of $m$-modes. We denote by $Q_j$ and $P_j$ the canonical position and momentum operators on the $j^{\text{th}}$ mode. The $j^{\text{th}}$ creation and annihilation operators $a_j=(Q_j+iP_j)/\sqrt{2}$ and $a_j^\dag=(Q_j-iP_j)/\sqrt{2}$ satisfy the well-known \textit{canonical commutation relations} (CCR):
\begin{align}
	\label{CCRlie}
	[a_j,a_k]=0\,,\qquad [a_j,a_k^\dagger]=\delta_{jk}I\,,
\end{align}
where $I$ denotes the identity operator on $\cH_m$. The Hilbert space $\cH_m$ can be understood as a projective representation of the symplectic group \cite{holevo2013quantum}. We denote by $\mathbb{M}_{2m}(\RR)$ the set of $2m \times 2m$ real matrices, and by $\operatorname{Sp}_{2m}(\RR)$, the set of symplectic matrices in $\mathbb{M}_{2m}(\RR)$, i.e.~matrices $S \in\mathbb{M}_{2m}(\RR)$ satisfying the condition $S\Omega_{m}S^{\operatorname{T}} = \Omega_{m}$, where $\Omega_{m}$ denotes the $2m\times 2m$ commutation matrix:
\begin{equation} \label{comm}
	\Omega_{m} \coloneqq \begin{pmatrix} 0 & -1 \\ 1 & 0 \end{pmatrix}^{\oplus m}\, .
\end{equation}
Any symplectic matrix $S$ has determinant equal to one and is invertible with $S^{-1} \in \operatorname{Sp}_{2m}(\RR)$. Hence, $\operatorname{Sp}_{2m}(\RR)$ is a subgroup of the special linear group $\operatorname{SL}_{2m}(\mathbb{R})$. In terms of the vector $R\coloneqq (Q_1,P_1,\dots,Q_m,P_m)$, the above relations take the compact form $[R_j,R_{k}]=-i (\Omega_{m})_{jk}$, where $\Omega_m$ denotes the $2m\times 2m$ standard symplectic form defined in \eqref{comm}. We will often omit the subscript $m$ if the number of modes is fixed. The \textit{total photon number} is defined by
\begin{equation}
	N\coloneqq \sum_{j=1}^m a_j^\dag a_j = \sum_{j=1}^m \frac{Q_j^2 + P_j^2}{2} - \frac{m}{2}\, .
	\label{total_photon_number}
\end{equation}
The operator $N$ is diagonal in the multi-mode Fock basis $\{|\mathbf{k}\rangle\}_{k\in \mathbb{N}^m}$, with 
\begin{equation}
	N\,|\mathbf{k}\rangle=\left(\sum_{i=1}^m k_i\right)\,|\mathbf{k}\rangle\,,\qquad \mathbf{k}\equiv (k_1,\dots, k_m)\,.
\end{equation}
We define the \textit{displacement operator} $\D(z)$ associated with a complex vector $z\in\CC^m$ as
\begin{equation}
	\D(z) = \exp \left[ \sum_j (z_j a^\dag_j - \overline{z}_j a_j) \right].
	\label{D}
\end{equation}
Thus, $\D(z)$ is a unitary operator and satisfies $\D(z)^\dag=\D(-z)$ and 
\begin{equation}
	\D(z) \D(w) = \D(z+w) \,e^{\frac12 (z^\intercal \overline{w} - z^\dag w)} ,
	\label{CCR Weyl}
\end{equation}
valid for all $z,w\in \CC^m$. We will also consider the real representation of the displacement operators: given $x=\in\mathbb{R}^{2m}$,
\begin{align*}
	V_x\coloneqq e^{iR^\intercal x}\equiv \D\left(\frac{-x_2+ix_1}{\sqrt{2}},\dots, \frac{-x_{2m}+ix_{2m-1}}{\sqrt{2}}\right)\,.
\end{align*}
A quantum state on $\cH_m$ is fully determined by its \textit{characteristic function} $\chi_\rho:\CC^{m}\to\CC$, given by
\begin{equation}
	\chi_\rho(z)\coloneqq \tr[\rho\,\D ( z)]\,.
	\label{chi}
\end{equation}
In what follows, we are going to use coherent states: given $\alpha=(\alpha_1,\dots,\alpha_m)\in\mathbb{C}^m$,
\begin{equation}
	|\alpha\rangle \equiv
	|\alpha_1,\dots,\alpha_m\rangle \coloneqq \D(\alpha)|0\rangle\,.
\end{equation}
The inner product between two single mode coherent states satisfies
\begin{equation*}
	\langle \beta|\alpha\rangle=e^{-\frac{|\alpha|^2+|\beta|^2}{2}+\alpha\overline{\beta}}\,.
\end{equation*}

In this article, we mainly consider a class of noises which includes the realistic description of the loss of a photon in each mode independently of the other, namely a tensor product of one-mode beamsplitters of transmissivity $\lambda\in (0,1)$ with arbitrary environment state $\sigma_E$:
\begin{align}\label{beamsplitter}
	\mathcal{N}_\lambda(\rho)=\tr_{E}\big[U_\lambda(\rho\otimes \sigma_E)U_\lambda^\dagger\big]\,,
\end{align}
where the beamsplitter unitary $U_\lambda\equiv \operatorname{exp}\big((a^\dagger b-b^\dagger a)\operatorname{arccos}(\sqrt{\lambda})\big)$ satisfies the relations
\begin{align}\label{eq:commutator}
	&    U_\lambda^\dagger a U_\lambda=\sqrt{\lambda}a+\sqrt{1-\lambda} b\,;\\
	&U_\lambda^\dagger b U_\lambda=-\sqrt{1-\lambda}a+\sqrt{\lambda} b\,.\nonumber
\end{align}
where $b$ stands for the annihilation operator associated to the environment mode. Note that we do not assume the environment state to be Gaussian. Although this is not required for our results to hold, we will fix the state $\sigma_E$ to be the same each time the noise $\cN_\lambda$ is applied, for sake of simplicity. For any subset $A\subseteq \Lambda$, we denote by $\cN_A\coloneqq\cN_\lambda^{\otimes |A|}$ the $|A|$-fold product of the channel $\cN_\lambda$ defined in \Cref{beamsplitter} acting on $A$.

\subsection{Gottesman-Kitaev-Preskill codes}\label{Gauss}

The Weyl commutation relations \eqref{CCR Weyl} imply the commutation of $\D(\alpha)$ and $\D(\beta)$ for any two complex numbers $\alpha,\beta$ up to a phase. In particular, if 
\begin{align}\label{alphabetacommute}
	\beta\overline{\alpha}-\overline{\beta}\alpha=i\pi\,,
\end{align}
the two operators anticommute, while if $\beta\overline{\alpha}-\alpha\overline{\beta}=2i\pi$, they commute. This leads to the choice of logical Pauli operators $\overline{X}=\D(\alpha)$ and $\overline{Z}=\D(\beta)$ where $\alpha$ and $\beta$ are any two complex numbers that satisfy \eqref{alphabetacommute}, which ensures that $\overline{X}\,\overline{Z}=-\overline{Z}\,\overline{X}$. We also define the logical Pauli operator $\overline{Y}\coloneqq \D(\alpha+\beta)=i\overline{X}\,\overline{Z}$. In order for the operators $\overline{X}$, $\overline{Y}$ and $\overline{Z}$ to behave like true Pauli operators, we then define the $\mathsf{GKP}$ logical codespace to be the simultaneous $+1$ eigenspace of the two operators 
\begin{align*}
	S_X\coloneqq \overline{X}^2=\D(2\alpha)\,,\qquad S_Z\coloneqq \overline{Z}^2=\D(2\beta)\,.
\end{align*}
This can be done since these operators commute. Therefore, the set $\{S_X^k,S_Z^l\}$ for $k,l\in\mathbb{Z}$ can be interpreted as the stabilizer group of the $\mathsf{GKP}$ code. The latter can be defined by choosing $\pm 1$ eigenstates of $\overline{Z}$. Unfortunately, the codewords are non-normalizable, which means that there is no physical process that can prepare a state lying exactly in the $\mathsf{GKP}$ codespace. In practice, one is required to resort to approximations of the latter. There exist various ways of defining approximate $\mathsf{GKP}$ states, and we refer to \cite{GrimsmoPuri2021} for a comprehensive review. For our purposes, we will not need an explicit expression for the codestates, and we simply refer to them as $|\widetilde{0}\rangle$ and $|\widetilde{1}\rangle$. The only characteristic of the state that we will need is that they have finite average energy with respect to the photon number operator $N\coloneqq a^\dagger a$, 
\begin{align*}
	\langle \widetilde{0}|N|\widetilde{0}\rangle\le E_0  ,\quad  \langle \widetilde{1}|N|\widetilde{1}\rangle\le E_0\,.
\end{align*}

Next, we discuss the ideal procedure for quantum error correction with the $\mathsf{GKP}$ code. To this end, we first define two generalized quadratures $\hat{Q}\coloneqq i(\overline{\beta}a-\beta a^\dagger)/\sqrt{\pi}$ and $\hat{P}\coloneqq -i(\overline{\alpha}a-\alpha a^\dagger)/\sqrt{\pi}$, such that $[\hat{Q},\hat{P}]=iI$. Here and for sake of simplicity, we stick to the so-called square code for which $\alpha=\sqrt{\pi/2}$ and $\beta=i\sqrt{\pi/2}$, and will often identify $Q=\hat{Q}$ and $P=\hat{P}$. Next, there are two canonical ways to perform $\mathsf{GKP}$ error correction using $\mathsf{GKP}$-encoded ancillae. These are continuous variables versions of Steane \cite{Steane97} and Knill \cite{Knill2005}. Here, we focus on Steane's method, see \Cref{fig:steane}. There, the entanglement gate is a $2$-mode CNOT gate, whose dual action on quadrature operators takes the form
\begin{align}
	(\mathcal{U}^{\operatorname{CNOT}}_{AA'})^{\dagger }&: Q_A\to Q_A,\,P_A\to P_A-P_{A'},\nonumber\\
	&\,\,\,\,\,Q_{A'}\to Q_A+Q_{A'},\,P_{A'}\to P_{A'}\,,\label{eq:CNOT}
\end{align}
where the control register is $A$.

\begin{figure}[h!]
	\centering
	\includegraphics[width=0.48\textwidth]{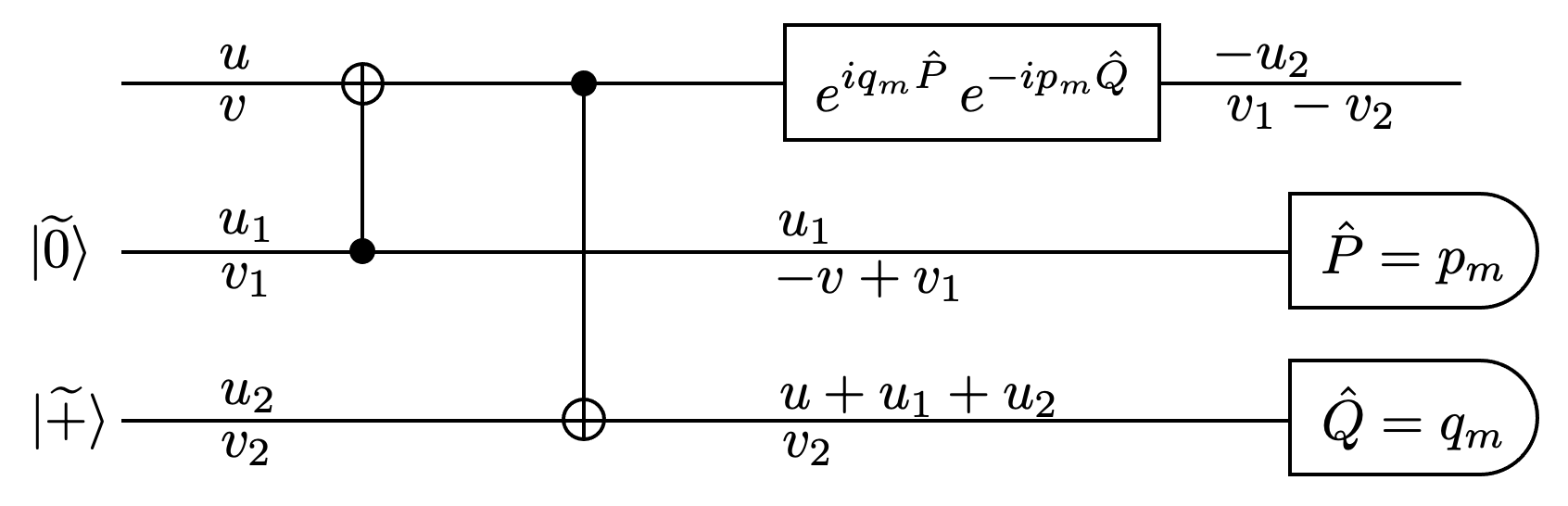}
	\caption{Illustration of Steane's error correction circuit. The labels $\{u,v\}$ next to a rail indicate a general displacement error $e^{-iu\hat{P}}e^{iv\hat{Q}}$, and the diagram indicates how the incoming errors propagate through the circuit. The two measurements are of the $\hat{P}$ and $\hat{Q}$ quadratures of the $\mathsf{GKP}$ code, respectively. For the correction shifts, we use the measurement outcomes modulo lattice spacing $\sqrt{\pi}$. Here, the approximate $\mathsf{GKP}$ state $|\widetilde{+}\rangle$ corresponds to the logical $|+\rangle$ state.}\label{fig:steane}
\end{figure}

In this paper, we also consider the concatenation of the $\mathsf{GKP}$ code with qubit stabilizer codes such as the toric code. For this, we need destructive logical measurements in any Pauli basis: 
\begin{align}\label{GKPlogicalpauli}
	\overline{X}=e^{-i\sqrt{\pi}\hat{P}},\qquad \overline{Z}=e^{i\sqrt{\pi}\hat{Q}},\qquad \overline{Y}=e^{i\sqrt{\pi}(\hat{Q}-\hat{P})}\,.
\end{align}
Then, Pauli measurements are performed by measuring one of three respective quadratures
\begin{align*}
	&    \mathcal{M}_X:\quad \text{ measure } -\hat{P}\\
	&    \mathcal{M}_Y:\quad \text{ measure } \hat{Q}-\hat{P}\\
	&  \mathcal{M}_Z:\quad \text{ measure } \hat{Q}\,,
\end{align*}
and rounding to the nearest multiple of $\sqrt{\pi}$. If the result is an even multiple, report a $+1$ outcome, and otherwise report a $-1$ outcome. The error correction schemes described so far include homodyne-type measurements of the quadratures $\hat{Q}$ and $\hat{P}$. However, the latter correspond to the infinite-squeezing limit of more realistic Gaussian POVM measurements as introduced in \cite{holevo2021structure}. For simplicity, we consider the unsharp (noisy) position and momentum measurements with variances $\alpha_q ,\alpha_p>0$ and $q,p\in\mathbb{R}$
\begin{align}\label{homodyneapprox}
	&m_{\hat{Q}}(q)=\frac{1}{2\pi\sqrt{\alpha}_q}V_{(0,-q)^t}\,e^{-\frac{1}{2\alpha_q}\hat{Q}^2}\,V_{(0,-q)^t}^\dagger=\frac{1}{2\pi\sqrt{\alpha_q}}e^{-\frac{1}{2\alpha_q}(\hat{Q}-q)^2}\,, \\
	& m_{\hat{P}}(p)=\frac{1}{2\pi\sqrt{\alpha}_p}V_{(p,0)^t}\,e^{-\frac{1}{2\alpha_p}\hat{P}^2}\,V_{(p,0)^t}^\dagger= \frac{1}{2\pi\sqrt{\alpha_p}}e^{-\frac{1}{2\alpha_p}(\hat{P}-p)^2}\,,\label{homodyneapproxbis}
\end{align}
and denote $\alpha_{\min}\coloneqq \min\{\alpha_p,\alpha_q\}$. The above operator valued densities are associated to POVMs $A\mapsto M(A)=\int_A m(x)\,dx$ for $m\in\{m_{\hat{P}},m_{\hat{Q}}\}$.

For our qubit stabilizer code, we consider a general situation where each of the $m$ modes corresponds to a vertex of a given graph $G=(V,E)$. Then, we examine a simple scheme of nondestructive measurement involving the introduction of one bosonic ancilla system per stabilizer, logical CNOT gates entangling modes associated with a stabilizer together with an ancilla, and logical Pauli measurements on the latter (see e.g.~\Cref{fig:meas} for the case of the 2D Toric code). This scheme was first considered in \cite{Dennis2002} for the qubit setting. Then, upon the measurement of all the stabilizers, the corresponding outcomes $s=\{s_k\}_{k=1}^K$ form a syndrome which is then used to construct an error correction unitary in the form of a displacement operator $D(f(s))$ for some classical function $f$ of the syndromes. We denote by $\Phi^{\operatorname{Stab}}_f$ the overall quantum error correction channel acting on the $m$ modes.

\begin{figure}[h!]
	\centering
	\includegraphics[width=0.48\textwidth]{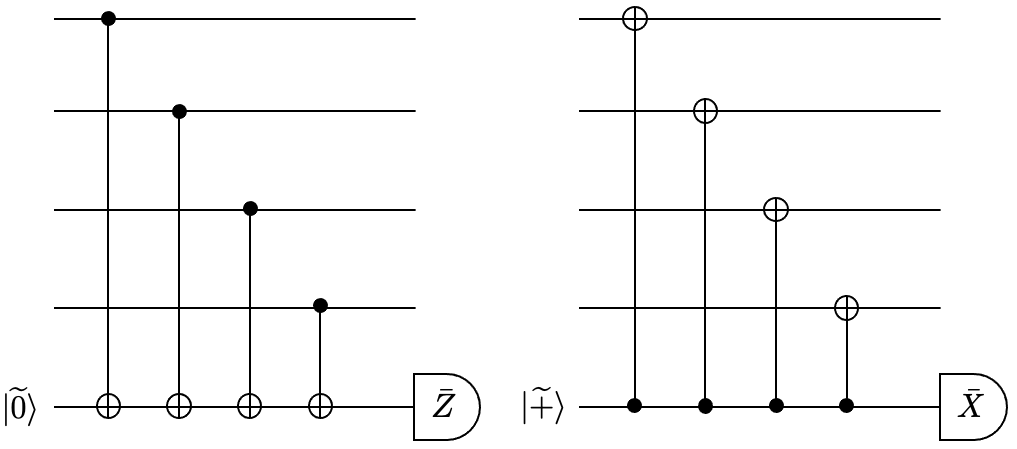}
	\caption{Ciruits for measurement of plaquette $\bar{Z}^{\otimes 4}$ and star $\bar{X}^{\otimes 4}$ stabilizers for the concatenated 2D Toric code.}\label{fig:meas}
\end{figure}

\section{Bosonic Wasserstein distance}\label{sec:proofnoisycircuit}

Our main technical contribution is the introduction and use of a non-commutative extension of the 
classical Lipschitz-type constant $\|\nabla f\|\coloneqq \max_{j\in[m]}\|\partial_j f\|_\infty$ of a real, continuously differentiable function $f$ of $m$ variables, and the dual notion of a Wasserstein-type distance between two distributions over $\mathbb{R}^m$.

\begin{definition}
	The bosonic Lipschitz constant of an operator $X\in \cB(\cH_m)$ is defined as
	\begin{equation}
		\|\nabla X\| ^2 \coloneqq\,\max_{j\in \Lambda}\,\max_{R_j\in\{Q_j,P_j\}}\,\sup_{R}\sup_{|\psi\rangle,|\varphi\rangle}\,|\langle \psi|[R_j,X]\otimes I_R|\varphi\rangle |^2 \equiv \max_j\,\|\nabla_j X\| \,,\nonumber
	\end{equation}
	where the supremum is over all pure states $|\psi\rangle,|\varphi\rangle\in\operatorname{dom}(\sqrt{N}\otimes I_R)$ for an arbitrarily large reference system $R$. By duality, we then define the \textit{bosonic Wasserstein distance} between two $m$-mode quantum states $\rho,\sigma\in \mathcal{D}(L_2(\mathbb{R}^m))$ as
	\begin{equation}\label{dualWasserstein1}
		W_{\operatorname{B}}(\rho,\sigma)\coloneqq\sup_{\|\nabla X\|\le 1}\,\Big|\tr\big[X(\rho-\sigma)\big]\Big|\,.
	\end{equation}
	where the supremum is over all bounded, self-adjoint operators $X$.
\end{definition}

In the next sections, we provide some basic properties of the bosonic Wasserstein distance $W_{\operatorname{B}}$.

\subsection{Relation to the trace distance}

First, we seek for an upper bound on the trace distance in terms of $W_{\operatorname{B}}$. By duality of both metrics, this amounts to finding an upper bound on the Lipschitz constant $\|\nabla X\|$ of any bounded operator $X$ in terms of its operator norm $\|X\|_\infty$. However, a bound of that sort does not exist (as classically, one can easily think of bounded observables which are not \textit{smooth}). In the classical setting, the problem can be handled by first \textit{smoothing} the function $f$, e.g. by convolving it with a Gaussian density $g$. In that case, one proves that there exists a finite constant $C>0$ such that $\|\nabla (f\ast g)\|\le C\|f\|_\infty$. Our first technical result consists in a bosonic analogue of this bound (see also \cite[Proposition 6.4]{gao2021ricci} for a single-mode variant).

\begin{prop}\label{regularity}
	For any two states $\rho_1,\rho_2\in\cD(L_2(\mathbb{R}^m))$ and $\lambda\in (0,1)$, we have 
	\begin{align}
		\|\cN^{\otimes m}_\lambda(\rho_1-\rho_2)\|_1\le \sqrt{\frac{\lambda}{1-\lambda}}\,\max\left\{\big\|[Q,\sigma_{E}]\big\|_1,\big\|[P,\sigma_{E}]\big\|_1\right\}\, W_{\operatorname{B}}(\rho_1,\rho_2)\,.
	\end{align}
	
\end{prop}
In order to prove \Cref{regularity}, we need the following technical lemma.

\begin{lemma}\label{techlemma}
	Given a two-mode bosonic system $AE$ with associated annihilation operators $a,b$, any reference system $C$, any $X\in\cB(L^2(\mathbb{R}))$ and any $|\varphi\rangle,|\psi\rangle\in\operatorname{dom}(\sqrt{N}\otimes I_C)$,
	\begin{align}
		&\langle \varphi|[a,\cN^\dagger_\lambda(X)]|\psi\rangle=\sqrt{\lambda}\,\langle\varphi| \cN^\dagger_\lambda ([a,X])|\psi\rangle= -\sqrt{\frac{\lambda}{1-\lambda}} \, \langle \varphi|\tr_E
		\left((I\otimes [b,\sigma]) U_{1-\lambda}(I\otimes X_{E})U_{1-\lambda}^\dagger\right)|\psi\rangle \label{eq:commutationrelationscnX}\\
		&\langle \varphi|[a^\dagger,\cN^\dagger_\lambda(X)]|\psi\rangle=\sqrt{\lambda}\,\langle\varphi| \cN^\dagger_\lambda ([a^\dagger,X])|\psi\rangle= -\sqrt{\frac{\lambda}{1-\lambda}} \, \langle \varphi|\tr_E
		\left((I\otimes [b^\dagger,\sigma]) U_{1-\lambda}(I\otimes X_{E})U_{1-\lambda}^\dagger\right)|\psi\rangle\label{eq:commutationrelationscnX1}
	\end{align}
	where $\cN_\lambda^\dagger$ is assumed to act on the system $A$ above, and we omitted the operator $I_C$ for ease of notations. 
\end{lemma}

\begin{proof}
	We only prove the first line \eqref{eq:commutationrelationscnX}, since \eqref{eq:commutationrelationscnX1} follows analogously with $a$ replaced by $a^\dagger$ and $b$ by $b^\dagger$. The adjoint map~$\cN^\dagger_\lambda$ is given by
	\begin{align}
		\cN^\dagger_\lambda(X)&=\tr_E \left((I\otimes \sigma) U_\lambda^\dagger (X\otimes I) U_\lambda\right)\  ,\label{eq:singlemodeintertwining}
	\end{align}
	and the first identity follows from this expression and Eq.~\eqref{eq:commutator} by a straightforward computation (see~\cite[Proof of Proposition 6.2]{gao2021ricci}). 
	The proof of~the second identity uses the identity     $U_\lambda^\dagger=U_{1-\lambda}e^{-\frac{\pi}{2}(a^\dagger b-b^\dagger a)}=
	U_{1-\lambda}\mathsf{SWAP}(I\otimes Z)$ where $\mathsf{SWAP}$ interchanges the two modes and $Z$ is a Gaussian unitary acting as~$ZbZ^\dagger=-b$, $Zb^\dagger Z^\dagger=-b^\dagger$ on the second system.    
	This implies that $U^\dagger_\lambda (X\otimes I)U_\lambda=U_{1-\lambda}(I\otimes X)U_{1-\lambda}^\dagger$ and thus 
	\begin{align}
		[a, \,\cN_\lambda^{\dagger}(X)]  &=\big[a,\tr_{E}\left((I\otimes \sigma_{E})U_\lambda^\dagger(X\otimes I_{E})U_\lambda\right)\big]\nonumber\\
		&=\left[a,\tr_{E}\left((I\otimes \sigma_{E})U_{1-\lambda}(I\otimes X)U_{1-\lambda}^\dagger\right)\right]\nonumber\\
		&=\tr_E\left( (I\otimes \sigma) \left[a\otimes I,
		U_{1-\lambda}(I\otimes X)U_{1-\lambda}^\dagger    \right]
		\right)\ ,\label{eq:intermediateanlambdax1}
	\end{align}
	where these identities hold when evaluated on states $|\varphi\rangle,|\psi\rangle\in\operatorname{dom}(\sqrt{N}\otimes I_C)$. Observe that
	\begin{align*}
		\left[a\otimes I,
		U_{1-\lambda}(I\otimes X)U_{1-\lambda}^\dagger    \right]
		&= \left[a\otimes I,U_{1-\lambda}\right](I\otimes X)U^\dagger_{1-\lambda}
		+U_{1-\lambda}(I\otimes X)\left[a\otimes I,U^\dagger_{1-\lambda}\right]
		+ U_{1-\lambda}[a\otimes I,I\otimes X]U_{1-\lambda}^\dagger
	\end{align*}
	where the third term on the right vanishes.
	Using the commutation relations~\eqref{eq:commutator} then gives
	\begin{align}
		\left[a\otimes I,
		U_{1-\lambda}(I\otimes X)U_{1-\lambda}^\dagger    \right]
		&= (1-\sqrt{1-\lambda}) \left[(a\otimes I), U_{1-\lambda}(I\otimes X)U_{1-\lambda}^\dagger\right]
		+\sqrt{\lambda}\left[ (I\otimes b), U_{1-\lambda}(I\otimes X) U_{1-\lambda}^\dagger\right]\nonumber\\
		&=(1-\sqrt{1-\lambda})
		\left[a\otimes I,    U^\dagger_\lambda (X\otimes I)U_\lambda\right]+\sqrt{\lambda}
		\left[ (I\otimes b), U_{1-\lambda}(I\otimes X) U_{1-\lambda}^\dagger\right]\ .\label{eq:aotimesidulambda1}
	\end{align}
	Inserting Eq.~\eqref{eq:aotimesidulambda1} into Eq.~\eqref{eq:intermediateanlambdax1}
	gives
	\begin{align*}
		[a,\,\cN_\lambda^{\dagger}(X)]&=(1-\sqrt{1-\lambda})[a,\cN_\lambda^\dagger(X)]+\sqrt{\lambda}\tr_E\left( (I\otimes \sigma)[I\otimes b,U_{1-\lambda}(I\otimes X) U_{1-\lambda}^\dagger] \right)\\
		&=(1-\sqrt{1-\lambda})[a,\cN_\lambda^\dagger(X)]-\sqrt{\lambda}\tr_E\left( (I\otimes [b,\sigma])U_{1-\lambda}(I\otimes X) U_{1-\lambda}^\dagger] \right)\ 
	\end{align*}
	as claimed. Here we used the tracial property in~\cite[Theorem 17]{brown1990jensen}. 
	The second identity  in Eq.~\eqref{eq:commutationrelationscnX} follows from this. 
\end{proof}

\begin{proof}[Proof of \Cref{regularity}]

	Eq.~\eqref{eq:commutationrelationscnX} immediately extends to several modes, giving for every $j\in \Lambda$ the identities
	\begin{align}
		\begin{matrix}
			[a_j,\cN^{\dagger\otimes m}_\lambda(X)]&=&\sqrt{\lambda} \cN^{\dagger\otimes m}_\lambda ([a_j,X])\\
			[a_j^\dagger,\cN^{\dagger\otimes m}_\lambda(X)]&=&\sqrt{\lambda} \cN^{\dagger\otimes m}_\lambda ([a_j^\dagger,X])\ 
		\end{matrix}
		\label{eq:multimodeintertwining}
	\end{align}
	because e.g., $[a_j,\cN^{\otimes m}(X)]=\bigotimes_{k\neq j}\cN_k\left([a_j, \cN_j(X)]\right)=
	\sqrt{\lambda}\left(\bigotimes_{k\neq j}\cN_k\right)\circ\cN_j\left([a_j,X]\right)=\sqrt{\lambda}\cN^{\otimes m}([a_j,X])$ where $\cN_k$ denotes the channel~$\cN=\cN_\lambda$ acting on the~$j$-th factor. We can similarly show that, for every $j\in \Lambda$,  
	\begin{align}
		\begin{matrix}
			[a_j,\cN_\lambda^{\dagger \otimes m}(X)]=-\sqrt{\frac{\lambda}{1-\lambda}}\,
			\cN^\dagger_{\Lambda\backslash \{j\}}\left(
			\tr_{E_j}\Big[(I\otimes [b_j,\sigma_{E_j}])U_{1-\lambda}(I\otimes X)U_{1-\lambda}^\dagger\Big]\right)\,,\\
			[a^\dagger_j ,\cN_\lambda^{\dagger \otimes m}(X)]=-\sqrt{\frac{\lambda}{1-\lambda}}\,
			\cN^\dagger_{\Lambda\backslash \{j\}}\left(
			\tr_{E_j}\Big[(I\otimes [b^\dagger_j,\sigma_{E_j}])U_{1-\lambda}(I\otimes X)U_{1-\lambda}^\dagger\Big]\right)\,.
		\end{matrix}\label{eq:toprovemultimodemodex}
	\end{align}
	where $b_j$ denotes the annihilation operator on the environment system~$E_j$. 
	Indeed, we have
	$ [a_j,\cN_\lambda^{\dagger \otimes m}(X)]=\cN^\dagger_{\Lambda\backslash \{j\}}\left(
	[a_j,\cN_j(X)]
	\right)$
	and similarly for~$a_j^\dagger$, hence Eq.~\eqref{eq:toprovemultimodemodex} is an immediate adaptation of the proof of~\eqref{techlemma}. Hence, for any two pure states $|\psi\rangle,\,|\varphi\rangle\in\operatorname{dom}(\sqrt{N}\otimes I_R)$, 
	\begin{align*}
		|\langle \psi|[Q_j,\cN_\lambda^{\dagger \otimes m}(X)]\otimes I_R|\varphi\rangle|&\le\sqrt{\frac{\lambda}{1-\lambda}}\,\sup_{\|y\|_1\le 1}\tr \big[(y\otimes [Q_j',\sigma_{E_j}])\,U_{1-\lambda}(I\otimes X)U_{1-\lambda}^\dagger\big]\\
		&\le \sqrt{\frac{\lambda}{1-\lambda}}\,\big\|[Q_j',\sigma_{E_j}]\big\|_1\,\|X\|_\infty\,,
	\end{align*}
	and similarly for $ |\langle \psi|\big[P_j,\cN_\lambda^{\dagger \otimes m}(X)\big]\otimes I_R|\varphi\rangle |$. The result follows by duality.
\end{proof}
\Cref{regularity} can be used to prove that the bosonic Wasserstein distance is a metric on quantum states:
\begin{cor}
	$W_{\operatorname{B}}:\cD(L_2(\mathbb{R}^m))\times \cD(L_2(\mathbb{R}^m))\to \mathbb{R}_+$ is a metric on the set of quantum states over $L_2(\mathbb{R}^m)$.
\end{cor}

\begin{proof}
	The triangle inequality and positivity are clear from the definition. Moreover $W_{\operatorname{B}}(\rho,\rho)=0$ for all $\rho\in \cD(L_2(\mathbb{R}^m))$. Finally, Given two states $\rho_1,\rho_2\in\cD(L_2(\mathbb{R}^m))$, \Cref{regularity} shows that $W_{\operatorname{B}}(\rho_1,\rho_2)=0$ implies $\cN_\lambda^{\otimes m}(\rho_1)=\cN_\lambda^{\otimes m}(\rho_2)$ for any $\lambda\in(0,1)$. Therefore, that $\rho_1=\rho_2$ follows from the strong continuity of $\lambda\mapsto \cN_\lambda^{\otimes m}$ with respect to the trace distance.
\end{proof}

\subsection{Diameter bound under moment constraints}\label{diameterboundsec}

Next, we prove an upper bound on the Wasserstein distance between two quantum states under a constraint on their moments.

\begin{prop}\label{W1toT1}
	For any two states $\rho,\sigma\in\cD(\cH_m)$ such that $\tr[\rho N],\tr[\sigma N]\le E<\infty$, 
	\begin{align*}
		W_{\operatorname{B}}(\rho,\sigma)\le 4 \sqrt{2\,m\,(m+E)}\,.
	\end{align*}
\end{prop}

Before proving \Cref{W1toT1}, we first state and prove a simple technical lemma which we will use multiple times in the section.
\begin{lemma}\label{lem.technical}
	Let $\Phi:\mathcal{T}_1(\cH_m)\to\mathcal{T}_1(\cH_{m'})$ be a moment-limited quantum channel, i.e.~for any  $\rho\in\cD(\cH_m)$ with $\tr[\rho N_m]<\infty$, $\tr[\Phi(\rho)N_{m'}]<\infty$. Then, for any reference system $R$ and any two pure states $|\varphi\rangle,|\psi\rangle\in\operatorname{dom}(\sqrt{{N}_m}\otimes I_R)$, the operator $A\coloneqq(\Phi\otimes \id_R)(|\psi\rangle\langle\varphi|)$ is trace-class, $\|A\|_1\le 1$, and $\operatorname{ran}(A), \operatorname{ran}(A^\dagger)\subseteq \operatorname{dom}(\sqrt{{N}_{m'}}\otimes I_R)$. In other words, $A$ admits the following representation
	\begin{align*}
		A=\sum_{\ell}\,\alpha_\ell \,|e_\ell\rangle\langle f_\ell|\,,
	\end{align*}
	for some sequences $\alpha\equiv \{\alpha_\ell\}$ of positive numbers with $\|\alpha\|_{\ell_1}\le 1$, and of orthonormal vectors $\{|e_\ell\rangle\}, \{|f_\ell\rangle\}$ in the domain of $\sqrt{{N}_{m'}}\otimes I_R$.
\end{lemma}

\begin{proof}
	Let $\{|e_i\rangle\}$ be an orthonormal basis in $\cH_R$ and let $|\chi\rangle\in \cH_{m'}\otimes \cH_R$. We have
	\begin{align*}
		\|{N}_{m'}^s\otimes I_R(\Phi\otimes \id_R)(|\psi\rangle\langle\varphi|)|\chi\rangle\|^2&=\sum_{k\in \mathbb{N}^{m'}}\sum_i\,E_k\,|\langle\chi|\,(\Phi\otimes \id_R)(|\varphi\rangle\langle \psi|)(|k\rangle\otimes |e_i\rangle)|^2\\
		&=\sum_{k\in\mathbb{N}^{m'}}\sum_i\,E_k\,|\langle\psi|(\Phi^\dagger\otimes \id_R)(|k\rangle\otimes |e_i\rangle\langle\chi|)\,|\varphi\rangle|^2\\
		&\le \sum_{k\in\mathbb{N}^{m'}}\sum_i\,E_k\,\langle\psi|(\Phi^\dagger\otimes \id_R)(|k\rangle\otimes |e_i\rangle\langle\chi|)(\Phi^\dagger\otimes\id_R)(|\chi\rangle\langle k|\otimes \langle e_i|)|\psi\rangle\\
		&\le \sum_{k\in\mathbb{N}^{m'}}\sum_i\,E_k\,\langle\psi| (\Phi^\dagger\otimes \id_R)(|k\rangle\langle k|\otimes |e_i\rangle\langle e_i|)|\psi\rangle\\
		&=\tr[(\Phi\otimes\id_R)(|\psi\rangle\langle\psi|)(N_{m'}\otimes I_R)]
	\end{align*}
	where $N_{m'}|k\rangle=E_k|k\rangle$, and where in the last inequality we have used the operator Schwarz inequality \cite[Theorem 5.3]{wolftour}. Now, since $\Phi$ is moment-limited and $|\psi\rangle\in\operatorname{dom}(\sqrt{{N}_m}\otimes I_R)$, we have proved that $\|\sqrt{{N}_{m'}}\otimes I_R(\Phi\otimes \id_R)(|\psi\rangle\langle\varphi|)|\chi\rangle\|<\infty$, and therefore $\operatorname{ran}(A)\subseteq\operatorname{dom}(\sqrt{{N}_{m'}}\otimes I_R)$. The same holds for $A^\dagger$ after exchanging the vectors $|\psi\rangle$ and $|\varphi\rangle$. Finally, since $\Phi$ is a channel, we have $\|\alpha\|_{\ell_1}=\|(\Phi\otimes \id_R)(|\psi\rangle\langle\varphi|)\|_1\le \||\psi\rangle\langle \varphi|\|_1=1$. 
\end{proof}

\begin{proof}[Proof of \Cref{W1toT1}]
	First, we denote the thermal Gaussian state at inverse temperature $\beta>0$ as
	\begin{align*}
		\sigma_\beta\coloneqq \,\frac{e^{-\beta\,N_\Lambda}}{\tr[e^{-\beta N_\Lambda}]}\,.
	\end{align*}
	The state $\sigma_\beta$ is the unique invariant state of the so-called Bose Ornstein-Uhlenbeck semigroup $(e^{t\cL_\beta})_{t\ge 0}$ of generator in the Heisenberg picture defined on a suitable subspace as~\cite{carlen2017gradient}
	\begin{align*}
		\cL_\beta^\dagger(X)\coloneqq \frac{1}{2}\sum_{j\in\Lambda}\left(\,e^{\frac{\beta}{2}}\,\Big(a_j^\dagger[X,a_j]+[a_j^\dagger,X]a_j\Big)+e^{-\frac{\beta}{2}}\,\Big( a_j[X,a_j^\dagger]+[a_j,X]a_j^\dagger\Big)\,\right).
	\end{align*}
	Next, we show an upper bound on the Wasserstein distance between an arbitrary state $\rho$ satisfying the constraints of the Proposition and $\sigma_\beta$, see Eq.~\eqref{eq:upperboundrhosigmabeta} below.  To this end, we proceed as follows. 
	Given a bounded, self-adjoint operator $X$ on $\cH_m$ satisfying the constraints in the definition of $W_{\operatorname{B}}$, we write $X_t\coloneqq e^{t\cL_\beta^\dagger}(X)$ and $\rho_t\coloneqq e^{t\cL_\beta}(\rho)$ such that 
	\begin{align}
		\tr[X(\sigma_\beta-\rho)]  &= \tr\left(X\lim_{t\rightarrow\infty} (e^{t\cL_\beta}(\rho)-\rho)\right)\nonumber \\
		&=\tr\left(X\int_0^\infty\cL_\beta(\rho_t) dt\right)\label{eq:xtintegral}\\
		&=\frac{1}{2}\,\sum_{j\in\Lambda}\int_0^\infty\,e^{\frac{\beta}{2}}\Big(\tr\big[[X_t,a_j]\rho a_j^\dagger\big]+\tr\big[a_j\rho [a_j^\dagger,X_t]\big]\Big)+e^{-\frac{\beta}{2}}\Big(\tr\big[[X_t,a_j^\dagger]\rho a_j\big]+\tr\big[a_j^\dagger\rho [a_j,X_t] \big]\Big)\,dt\,.\nonumber 
	\end{align}
	The maps $e^{t\cL_\beta}$ turn out to be equal to the channels $\cN_\lambda^{\otimes m}$ for $\lambda\equiv e^{-2\operatorname{sinh}(\beta/2)t}$ and environment states $\sigma_{E_j}=\sigma_\beta$. Therefore, by the  relations~\eqref{eq:multimodeintertwining} and the definition of~$X_t$ we have 
	\begin{align}
		\tr\left([X_t,a_j]\rho a_j^\dagger\right)
		&=e^{-\operatorname{\sinh}(\beta/2)t}\tr\left(e^{t\cL_\beta^\dagger}([X,a_j])\rho a_j^\dagger\right)=e^{-\operatorname{\sinh}(\beta/2)t}\tr\left([X,a_j]e^{t\cL_\beta}(\rho a_j^\dagger)\right)\ \label{eq:firstidentitycommutat}
	\end{align}
	and similarly
	\begin{align}
		\begin{matrix}
			\tr\left(a_j\rho[a_j^\dagger,X_t]\right)&=e^{-\operatorname{\sinh}(\beta/2)t}
			\tr\left(e^{t\cL_\beta}(a_j\rho)[a_j^\dagger,X]\right)\\
			\tr\left([X_t,a_j^\dagger]\rho a_j\right)&=e^{-\operatorname{\sinh}(\beta/2)t}
			\tr\left([X,a_j^\dagger] e^{t\cL_\beta}(\rho a_j)\right)\\
			\tr\left(a_j^\dagger \rho[a_j,X_t]\right)&=e^{-\operatorname{\sinh}(\beta/2)t}
			\tr\left(e^{t\cL_\beta}(a_j^\dagger\rho)[a_j,X]\right)
		\end{matrix}\ \quad .\label{eq:secondidentitiescommutatm}
	\end{align}
	Inserting~\eqref{eq:firstidentitycommutat} and~\eqref{eq:secondidentitiescommutatm} into~\eqref{eq:xtintegral} we obtain
	\begin{align}
		& \tr[X(\sigma_\beta-\rho)]=\frac{1}{2}\int_0^\infty e^{-\operatorname{sinh}(\beta/2)t}\,\sum_{j\in\Lambda}\,e^{\frac{\beta}{2}}\left(
		\tr\left([X,a_j]e^{t\cL_\beta}(\rho a_j^\dagger)\right)
		+   \tr\left(e^{t\cL_\beta}(a_j\rho)[a_j^\dagger,X]\right)\right)
		\label{fourtraces}\\
		&\qquad\qquad\qquad\qquad\quad\qquad\qquad\qquad+e^{-\frac{\beta}{2}}\left(
		\tr\left([X,a_j^\dagger] e^{t\cL_\beta}(\rho a_j)\right)+\tr\left(e^{t\cL_\beta}(a_j^\dagger\rho)[a_j,X]\right)
		\right)\,dt\,.\nonumber
	\end{align}
	Next, we control each of the traces above separately. For this, we consider the spectral decomposition $\rho\coloneqq \sum_{\ell}\lambda_\ell |\psi_\ell\rangle\langle\psi_\ell|$. We first consider the operator $e^{t\cL_\beta }(|\psi_\ell\rangle\langle \varphi_\ell| )$ for some arbitrary eigenvector $|\psi_\ell\rangle$ of $\rho$ such that $a_j|\psi_\ell\rangle\ne 0$ and $|\varphi_\ell\rangle\coloneqq a_j|\psi_\ell\rangle /\|a_j|\psi_\ell\rangle\|$. Then, by assumption on $\rho$, $|\psi_\ell\rangle,|\varphi_\ell\rangle\in\operatorname{dom}(\sqrt{N})$. Moreover, since $e^{t\cL_\beta}$ is Gaussian for all $t\ge 0$, we have that for any state $\rho'$ such that $\tr[\rho'N]<\infty$, $\tr[e^{t\cL_\beta}(\rho')N]<\infty$. By \Cref{lem.technical}, the operator $e^{t\cL_\beta}(|\psi_\ell\rangle\langle\varphi_\ell|)$ admits a decomposition onto rank-one operators of the form
	\begin{align*}
		e^{t\cL_\beta}(|\psi_\ell\rangle\langle\varphi_\ell|)=\sum_i \alpha^{(\ell)}_i\,|e^{(\ell)}_i\rangle\langle f^{(\ell)}_i|\,,\qquad\text{ with }\qquad  \|\alpha^{(\ell)}\|_{\ell_1}\le 1\,,\quad \{|e^{(\ell)}_i\rangle,|f^{(\ell)}_i\rangle\}\in\operatorname{dom}(\sqrt{N})\,.
	\end{align*}
	Moreover, we have
	\begin{align}
		\tr\left(e^{t\cL_\beta^\dagger} ([X,a_j]) \rho a_j^\dagger\right)
		&=\tr\left([X,a_j]e^{t\cL_\beta}(\rho a_j^\dagger)\right)\\
		&=\sum_{\ell}\lambda_\ell \|a_j\ket{\psi_\ell}\|\cdot \tr\left(
		[X,a_j] e^{t\cL_\beta}(\ket{\psi_\ell}\bra{\varphi_\ell})\right)\\
		&=\sum_{\ell,i}\lambda_\ell\,\alpha^{(\ell)}_i\,\|a_j|\psi_\ell\rangle\|\,|\langle f_i^{(\ell)}|[X,a_j]|e_i^{(\ell)}\rangle|\\
		&\le \sqrt{2} \,\|\nabla X\|\,\sum_{\ell}\,\lambda_\ell\,\|a_j|\psi_\ell\rangle\|\\
		&\le \sqrt{2}\,\|\nabla X\|\,\tr(\rho N_j)^{\frac{1}{2}}\,, 
	\end{align}
	where the last inequality follows by Jensen's inequality. Therefore, for all $t\ge 0$,
	\begin{align}
		\big|\tr[e^{t\cL_\beta^\dagger}([X,a_j])\rho a_j^\dagger]\big|\le \sqrt{2}\,\|\nabla X\|\,\tr(\rho N_j)^{\frac{1}{2}}\,.\label{eq:firstupperboundxaj}
	\end{align}
	Similarly, we find that 
	\begin{align}
		\big|\tr[a_j\rho\,e^{t\cL_\beta^\dagger}([a_j^\dagger,X ])]\big|\le \sqrt{2}\,\|\nabla X\|\,\tr(\rho N_j)^{\frac{1}{2}}\,.\label{eq:secondupperboundxaj}
	\end{align}
	and 
	\begin{align}
		\big|\tr[e^{t\cL_\beta^\dagger}([X,a_j^\dagger ])\rho a_j]\big|,\,   \big|\tr[a_j^\dagger \rho\,e^{t\cL_\beta^\dagger}([a_j,X ])]\big|\le \sqrt{2}\,\|\nabla X\|\,\tr(\rho (I+N_j))^{\frac{1}{2}}\,.\label{eq:thirdupperboundxaj}
	\end{align}
	Therefore, inserting~\eqref{eq:firstupperboundxaj},~\eqref{eq:secondupperboundxaj} and
	~\eqref{eq:thirdupperboundxaj} into Eq.~\eqref{fourtraces}
	and integrating over~$t$ gives 
	\begin{align*}
		\big|\tr[X(\sigma_\beta-\rho)]\big|&\le \frac{\sqrt{2}\,\|\nabla X\|}{\operatorname{sinh}(\beta/2)}\,\sum_{j\in\Lambda}\, \left(e^{\frac{\beta}{2}}\,\tr(\rho N_j)^{\frac{1}{2}}+e^{-\frac{\beta}{2}}\, \tr(\rho(I+N_j))^{\frac{1}{2}}\right)\\
		&\le \frac{2\sqrt{2}\,\,\|\nabla X\|\operatorname{cosh}(\beta/2)}{\operatorname{sinh}(\beta/2)}\,\sum_{j\in\Lambda}\tr(\rho (I+N_j))^{\frac{1}{2}}\\
		&\le \frac{2\sqrt{2}\,\,|\Lambda|^{\frac{1}{2}}\|\nabla X\|\operatorname{cosh}(\beta/2)}{\operatorname{sinh}(\beta/2)}\,\Big(\sum_{j\in\Lambda}\tr(\rho(I+N_j))\Big)^{\frac{1}{2}}\\
		&= \frac{2\sqrt{2}\,\,|\Lambda|^{\frac{1}{2}}\|\nabla X\|\operatorname{cosh}(\beta/2)}{\operatorname{sinh}(\beta/2)}\,\Big(|\Lambda|+\tr(\rho N)\Big)^{\frac{1}{2}}\,.
	\end{align*}
	From duality, i.e., expression~\eqref{dualWasserstein1}, it follows immediately that
	\begin{align}
		W_{\operatorname{B}}(\rho,\sigma_\beta)\le \frac{2\sqrt{2}\,|\Lambda|^{\frac{1}{2}}\operatorname{cosh}(\beta/2)}{\operatorname{sinh}(\beta/2)}\,\Big(|\Lambda|+\tr(\rho N)\Big)^{\frac{1}{2}}\,.\label{eq:upperboundrhosigmabeta}
	\end{align}
	Applying the same reasoning to~$\sigma$ instead of~$\rho$ and using the triangle inequality implies
	\begin{align*}
		W_{\operatorname{B}}(\rho,\sigma)\le \frac{4\sqrt{2}\,|\Lambda|^{\frac{1}{2}}\operatorname{cosh}(\beta/2)}{\operatorname{sinh}(\beta/2)}\,\Big(|\Lambda|+\tr(\rho N)\Big)^{\frac{1}{2}}\,.
	\end{align*}
	The result then follows  optimizing over~$\beta$, i.e., taking the limit~$\beta\rightarrow\infty$.
\end{proof}

\Cref{W1toT1} can be easily refined in the case where the states $\rho$ and $\sigma$ have coinciding marginals on a subset of the modes.

\begin{prop}
	Let $A\subset [m]$ and $\rho,\sigma\in\cD(\cH_m)$ be such that $\rho_A=\sigma_A$ and $\tr[\rho N_{A^c}], \tr[\sigma N_{A^c}]\le E<\infty$ Then 
	\begin{align*}
		W_{\operatorname{B}}(\rho,\sigma)\le 4\sqrt{2|A^c|(|A^c|+E)}\,.
	\end{align*}
\end{prop}
\begin{proof}
	The proof is similar to tat of \Cref{W1toT1}. Instead, we consider the decomposition $W_{\operatorname{B}}(\rho,\sigma)\le W_{\operatorname{B}}(\rho,\rho_A\otimes \sigma_\beta)+W_{\operatorname{B}}(\sigma_A\otimes \sigma_\beta,\sigma)$, where $\sigma_\beta$ is the invariant state of the quantum Ornstein-Uhlenbeck semigroup acting on the subsystem $A^c$. 
\end{proof}

\section{Bosonic Wasserstein contraction coefficients}

The proof of the main result of this work relies on the analysis of the Wasserstein distance between two arbitrary states passing through layers of the noisy error correction protocol. For this, we need to relate the Wasserstein distance at the output of a gate included in the latter in terms of the Wasserstein distance at its input. This notion is well captured by the concept of a contraction coefficient:
\begin{definition}
	The bosonic Lipschitz contraction coefficient of a quantum channel $\Phi:\cT_1(\cH_m)\to \cT_1(\cH_m)$ is defined as
	\begin{align*}
		\|\Phi^\dagger\|_{\nabla\to\nabla}\coloneqq \sup_{\|\nabla X\|\le 1}\,\|\nabla \Phi^\dagger(X)\|\,,
	\end{align*}
	where the optimization is over all bounded, self-adjoint operators $X$. Similarly, the bosonic Wasserstein contraction coefficient of $\Phi$ is defined as 
	\begin{align*}
		\|\Phi\|_{W_{\operatorname{B}}\to W_{\operatorname{B}}}\coloneqq \sup_{\rho\ne\sigma\in\cD(\cH_m)}\frac{W_{\operatorname{B}}(\Phi(\rho),\Phi(\sigma))}{W_{\operatorname{B}}(\rho,\sigma)}\,.
	\end{align*}
	
\end{definition}
By duality of $W_{\operatorname{B}}$ and $\|\nabla(.)\|$, we clearly have that
\begin{align}\label{eq:duality}
	\|\Phi\|_{W_{\operatorname{B}}\to W_{\operatorname{B}}}\le \|\Phi^\dagger\|_{\nabla\to\nabla}\,.
\end{align}
Indeed, for any $X$ with $\|\nabla X\|\leq 1$, we have that $Y\coloneqq \frac{1}{\|\Phi^\dagger\|_{\Delta\to\Delta}}\Phi^\dagger(X)$ satisfies $\|\nabla Y\|\leq 1$, hence 
\begin{align*}
	|\tr\left(\Phi^\dagger(X)(\rho_1-\rho_2)\right)|=\|\Phi^\dagger\|_{\nabla\to\nabla}\,|\tr(Y(\rho_1-\rho_2))|\leq \|\Phi^\dagger\|_{\nabla\to\nabla} \,W_{\operatorname{B}}(\rho_1,\rho_2)
\end{align*}
and the claim \eqref{eq:duality} follows by taking the supremum over~$X$.

In the next proposition, we show that the bosonic Lipschitz contraction coefficient tensorizes:

\begin{prop}\label{tensorization}
	Given $\Phi=\bigotimes_{j=1}^m\Phi_j$ a tensor product of moment-limited quantum channels $\Phi_j:\cT_1(\cH_1)\to \cT_1(\cH_1)$. Then
	\begin{align*}
		\|\Phi^\dagger\|_{\nabla\to\nabla }\le  \max_{j\in [m]}\,\|\Phi^\dagger_j\otimes \id_{[m]\backslash \{j\}}\|_{\nabla_j\to\nabla_j } \,,
	\end{align*}
	where $\|\Phi^\dagger_j\otimes \id_{[m]\backslash \{j\}}\|_{\nabla_j\to\nabla_j }\coloneqq \sup_{\|\nabla_j X\|\le 1}\,\|\nabla_j \Phi^\dagger_j\otimes \id_{[m]\backslash \{j\}}(X)\|$.
\end{prop}

\begin{proof}
	Let $j\in[m]$, $R$ be an arbitrary reference system and $|\varphi\rangle,|\psi\rangle\in\operatorname{dom}(\sqrt{N}\otimes I_R)$. Since $\Phi_k$ is moment-limited for each $k\in[m]$, we have by \Cref{lem.technical} that there exist two sequences of orthogonal, normalized vectors $\{|e_i\rangle\}_{i\in I}$ and $\{|f_i\rangle\}_{i\in I}$ in the domain of $\sqrt{N}\otimes I_R$, as well as a sequence $\alpha=\{\alpha_i\}_{i\in I}$ of positive numbers in $\ell_1$ with $\| \alpha\|_{\ell_1}\le 1$ such that $(\Phi_{\Lambda\backslash \Lambda_j}\otimes \id_{\Lambda_jR})(|\psi\rangle\langle \varphi|)=\sum_i\,\alpha_i|e_i\rangle\langle f_i|$, where we denote $\Lambda_j\coloneqq \{j\}$. Therefore, for any self-adjoint bounded $X$ and $R_j\in\{P_j,Q_j\}$,
	\begin{align*}
		|\langle \varphi|\big[R_j,\Phi^\dagger(X)\big]|\otimes I_R|\psi\rangle|\!=\!\left|\sum_i\alpha_i\langle f_i|[R_j,\Phi_j^\dagger(X)]\otimes I_R|e_i\rangle\right|\overset{(1)}{\le}\! \sup_{R',|\psi'\rangle,|\varphi'\rangle}\!|\langle \varphi'|\big[R_j,\Phi_j^\dagger(X)\big]\otimes I_{R'}|\psi'\rangle|
		\!\le   \!\|\nabla_j \Phi_j^\dagger(X)\|\,,
	\end{align*}
	where the inequality $(1)$ follows from choosing $R'=RI$, $|\psi'\rangle=\sum_i\sqrt{\alpha_i}|e_i\rangle\otimes |i\rangle $ and $|\varphi'\rangle=\sum_i\sqrt{\alpha_i}|f_i\rangle\otimes |i\rangle $, for some orthonormal basis $\{|i\rangle\}_{i\in I}$. We have therefore that
	\begin{align*}
		\|\nabla \Phi^\dagger(X)\|\le \max_{j\in [m]}\,\|\nabla_j\Phi_j^\dagger(X)\|\le \max_{j\in [m]}\,\|\Phi^\dagger_j\otimes \id_{[m]\backslash \{j\}}\|_{\nabla_j\to\nabla_j }\,\,\|\nabla X\|\,
	\end{align*}
	and the result follows.
\end{proof}
In the next sections, we provide estimates for the bosonic Wasserstein contraction coefficients for the channels appearing in the approximate $\mathsf{GKP}$ error correction scheme described in \Cref{Gauss}.

\subsection{Contraction under the noise channel $\cN_\lambda$}
We start by showing that the bosonic Wasserstein distance contracts under the 

\begin{prop}\label{prop:contractionNlambda}For any $\lambda\in [0,1]$,
	\begin{align}
		\|\cN_\lambda^{\otimes m}\|_{W_{\operatorname{B}}\to W_{\operatorname{B}}} \le \sqrt{\lambda}\,.\label{NlW1}
	\end{align}
\end{prop}

\begin{proof}
	By \eqref{eq:duality}, it is enough to show that $\|\nabla  \cN_\Lambda^\dagger(X)\|\le \sqrt{\lambda}\|\nabla  X\|$ for any bounded, Lipschitz observable $X$. For any reference system $R$, any $j\in \Lambda$ and all $|\psi\rangle,|\varphi\rangle \in\operatorname{dom}(\sqrt{N}\otimes I_R)$, we have
	\begin{align*}
		\langle\varphi| [Q_j, (\cN_\Lambda)^\dagger (X)]\otimes I_R|\psi\rangle&= \tr[(\cN_\Lambda)^\dagger(X)\otimes I_R|\psi\rangle\langle \varphi| Q_j\otimes I_R]-\tr[Q_j \otimes I_R|\psi\rangle\langle \varphi| (\cN_\Lambda)^\dagger (X)\otimes I_R]\\
		&=\tr[X\otimes I_R(\cN_\Lambda\otimes \id_R)(|\psi\rangle\langle\varphi| Q_j\otimes I_R)]-\tr[(\cN_\Lambda\otimes \id_R)(Q_j\otimes I_R|\psi\rangle\langle\varphi|)X\otimes I_R]\\
		&=\sqrt{\lambda}\,\tr\big[(\cN_\Lambda\otimes \id_R)(|\psi\rangle\langle \varphi|)[Q_j,X]\otimes I_R\big]\,,
	\end{align*}
	where the last identity follows from the commutation relations \eqref{eq:commutator}. The same holds for $\langle \varphi| [P_j,(\cN_\Lambda)^{\dagger}(X)]\otimes I_R|\psi\rangle$. Now, since $\cN_\Lambda$ is moment-limited (see \cite{Shirokov2020}), we have by \Cref{lem.technical} that there exist two sequences of orthogonal, normalized vectors $\{|e_i\rangle\}$ and $\{|f_i\rangle\}$ in the domain of $\sqrt{N}\otimes I_R$, as well as a sequence $\alpha=\{\alpha_i\}$ of positive numbers in $\ell_1$ with $\| \alpha\|_{\ell_1}\le 1$ such that $(\cN_\Lambda\otimes \id_R)(|\psi\rangle\langle \varphi|)=\sum_i\,\alpha_i|e_i\rangle\langle f_i|$. Therefore, we have
	\begin{align}
		\|\nabla \cN_\Lambda^\dagger(X)\|\le\,\sqrt{\lambda}\max_{j} \max_{R_j\in\{Q_j,P_j\}} \sup_R\sup_{|\psi\rangle,|\varphi\rangle}  \sum_i\,\alpha_i\,|\langle f_i|[R_j,X]\otimes I_R|e_i\rangle|
		\le \sqrt{\lambda}\,\|\nabla  X\|\,.\label{eq:upperboundlambdanableX}
	\end{align}

\end{proof}

\subsection{Boundedness under $\mathsf{GKP}$ error correction}

In this section, we control the growth of the Wasserstein distance due to the error correction schemes for the $\mathsf{GKP}$ code. We start by Steane's quantum error correction procedure which is illustrated in \Cref{fig:steane}: when acting on a mode $j\in\Lambda$, it corresponds to the map
\begin{align*}
	\Phi^{\operatorname{St}}_j(\rho_\Lambda)= \int  \tr_{B_jC_j}\Big[\big(I_{\Lambda}\otimes m_{\hat{P}}(p)_{B_j}\otimes m_{\hat{Q}}(q)_{C_j}\big)\mathcal{V}^{(q,p)}_{\Lambda_j}\Big(  \mathcal{U}^{\operatorname{CNOT}}_{\Lambda_jC_j}\circ \mathcal{U}^{\operatorname{CNOT}}_{B_j\Lambda_j} (\rho_\Lambda\otimes \sigma_{B_jC_j})\Big)\Big]\,dq\,dp
\end{align*}
where we write $\Lambda_j\coloneqq \{j\}$, $\mathcal{V}^{(q,p)}\coloneqq e^{iq{P}}e^{-ip{Q}}(.) e^{ip{Q}}e^{-iq{P}}$ and $\sigma_{B_jC_j}\coloneqq  |\widetilde{0}\rangle\langle \widetilde{0}|_{B_j}\otimes |\widetilde{+}\rangle\langle \widetilde{+}|_{C_j}$. Using the abbreviations~$\mathcal{U}_{\Lambda_jB_jC_j}\coloneqq \mathcal{U}^{\operatorname{CNOT}}_{\Lambda_jC_j}\circ \mathcal{U}^{\operatorname{CNOT}}_{B_j\Lambda_j}$ and $m(q,p)\coloneqq  m_{\hat{P}}(p)\otimes m_{\hat{Q}}(q)$, we have
\begin{align}
	\Phi^{\operatorname{St}}_j(\rho_\Lambda)= \int  \tr_{B_jC_j}\Big[\big(I_{\Lambda}\otimes m(q,p)_{B_jC_j}\big)\mathcal{V}^{(q,p)}_{\Lambda_j} \mathcal{U}_{\Lambda_jB_jC_j}(\rho_\Lambda\otimes \sigma_{B_jC_j})\Big)\Big]\,dq\,dp\ .
\end{align}

\begin{prop}\label{prop:ECW1}
	The quantum channel $\Phi^{\operatorname{St}}$ satisfies
	\begin{align*}
		\|(\Phi^{\operatorname{St}})^{\otimes m}\|_{W_{\operatorname{B}}\to W_{\operatorname{B}}}\le 2\,.
	\end{align*}
\end{prop}
To prove \Cref{prop:ECW1}, we argue once again by duality. We need the adjoint map $\Phi^{\operatorname{St}}_j$ 
\begin{align}
	(\Phi^{\operatorname{St}}_j)^\dagger(X)&=\int \tr_{B_jC_j}\left(
	(I_\Lambda\otimes \sigma_{B_jC_j})\cU_{\Lambda_jB_jC_j}^\dagger((\mathcal{V}^{(q,p)}_{\Lambda_j})^\dagger  (X)\otimes m(q,p)_{B_jC_j})\right)
	dq\,dp\,\nonumber\\
	&= \tr_{B_jC_j}\left((I_\Lambda\otimes \sigma_{B_jC_j}) 
	\cU_{\Lambda_jB_jC_j}^\dagger \Phi_j^\dagger(X)
	\right)\label{eq:adjointmapexplicit}
\end{align}
where we introduced the completely positive  unital map  
\begin{align*}
	\Phi^\dagger_j:X\mapsto \, \int\,(\mathcal{V}_{\Lambda_j}^{(q,p)})^\dagger (X)\otimes m(q,p)_{B_jC_j}\,dq\,dp\,.
\end{align*}
By \Cref{tensorization}, it is enough to consider the maps $\Phi^{\operatorname{St}\dagger}_j$ and commutations with $P_j$ and $R_j$. Inserting the expression~\eqref{eq:adjointmapexplicit} for $\Phi^{\operatorname{St}\dagger}$ and using the abbreviation
\begin{align*}
	Y_{\Lambda B_jC_j}&\coloneqq \cU_{\Lambda_jB_jC_j}^\dagger \Phi_j^\dagger(X)\,,\end{align*}
this takes the form
\begin{align*}
	\langle \varphi|\left(\left[P_j,\Phi^{\operatorname{St}\dagger}_j(X)\right]\otimes I_R\right)\ket{\psi}&=\tr\left(\left[P_j,\tr_{B_jC_j}\left((I_\Lambda\otimes\sigma_{B_jC_j})Y_{\Lambda B_jC_j}\right)\right]\otimes I_R \,\ket{\psi}\bra{\varphi}\right)\\
	&=\tr\left(\left([P_j\otimes I_{B_jC_j},Y_{\Lambda B_jC_j}]\otimes I_R\right)(\sigma_{B_jC_j}\otimes \ket{\psi}\bra{\varphi})\right)
\end{align*}
where we used that $\tr_{B_jC_j}\left((I_\Lambda\otimes\sigma_{B_jC_j})Y_{\Lambda B_jC_j}\right)=\tr_{B_jC_j}\left(Y_{\Lambda B_jC_j}(I_\Lambda\otimes\sigma_{B_jC_j})\right)$ in the second identity. That is, we have 
\begin{align}
	\langle \varphi|\left(\left[P_j,\Phi^{\operatorname{St}\dagger}_j(X)\right]\otimes I_R\right)\ket{\psi}=
	\tr\left(\left(\left[P_j\otimes I_{B_jC_j},\cU^\dagger_{\Lambda_jB_jC_j}\Phi_j^\dagger(X)\right]\otimes I_R\right)(\sigma_{B_jC_j}\otimes 
	\ket{\psi}\bra{\varphi})
	\right)\nonumber\\
	=\tr\left(\left(\left[\cU_{\Lambda_jB_jC_j}(P_j\otimes I_{B_jC_j}),\Phi_j^\dagger(X)\right]\otimes I_R\right)
	\cU_{\Lambda_jB_jC_j}
	(\sigma_{B_jC_j}\otimes 
	\ket{\psi}\bra{\varphi})
	\right)\nonumber\\
	=\tr\left(\left(\left[P_j\otimes I_{B_jC_j}-P_{C_j}\otimes I_{\Lambda B_j},\Phi_j^\dagger(X)\right]\otimes I_R\right)
	\cU_{\Lambda_jB_jC_j}
	(\sigma_{B_jC_j}\otimes 
	(\ket{\psi}\bra{\varphi})
	\right)\label{eq:expansionexprone}
\end{align}
where we applied the identity~
\begin{align}
	\tr\left([A,\cU^\dagger(B)]C\right)=\tr\left([\cU(A),B)]\cU(C)\right)\label{eq:cyclicunitaryprop}
\end{align} for conjugation~$\cU(\cdot)=U\cdot U^\dagger$ by a unitary~$U$ in the second step, and used the definition of~$\cU_{\Lambda_jB_jC_j}$ (and the corresponding action of~$\operatorname{CNOT}$ defined in \Cref{eq:CNOT}) in the last step.

\begin{proof}[Proof of \Cref{prop:ECW1}]
	By \eqref{eq:duality} and \Cref{tensorization}, it is enough to show that $\|\Phi^{\operatorname{St}\dagger}_j\otimes \id_{[m]\backslash \{j\}}\|_{\nabla_j\to\nabla_j }\le 2$ for each $j\in[m]$. For sake of clarity, we will omit the identity $\id_{[m]\backslash \{j\}}$ in the following. Let $X$ be a self-adjoint, bounded operator on $\cH_m$ satisfying the constraints in the definition of $W_{\operatorname{B}}$ and fix a reference system $R$. For any two states $|\psi\rangle,|\varphi\rangle\in\operatorname{dom}(\sqrt{N}\otimes I_R)$, we denote $  (U_{\Lambda_jB_jC_j}\otimes I_R)\sigma_{B_jC_j}\otimes |\psi\rangle\langle\varphi|(U_{\Lambda_jB_jC_j}\otimes I_R)^{\dagger }=|\psi'\rangle\langle \varphi'|$, where $|\psi'\rangle$ and $|\varphi'\rangle$ are in the domain of $\sqrt{N_{B_jC_j}+N_{\Lambda}}\otimes I_R$ \cite{Shirokov2020}. Therefore,  by~\eqref{eq:expansionexprone}, we obtain
	\begin{align}
		\langle\varphi|[P_j,\Phi^{\operatorname{St}\dagger}_j (X)]\otimes I_R|\psi\rangle&=\Delta_j-\Delta_j'\label{eq:varphipsipjPhist}
	\end{align}
	where
	\begin{align}
		\Delta_j&= \langle \varphi'| ([P_j\otimes I_{B_jC_j},\Phi_j^\dagger(X)]\otimes I_R)|\psi'\rangle\label{eq:deltajdef}\\
		\Delta'_j&=\langle \varphi'| ([P_{C_j}\otimes I_{\Lambda B_j},\Phi_j^\dagger(X)]\otimes I_R)|\psi'\rangle\label{eq:deltajprimedef}\,.
	\end{align}

	We first analyze the expression~$\Delta_j$. Observe that, again by~\eqref{eq:cyclicunitaryprop} we have that
	\begin{align*}
		\langle \varphi'| ([P_j\otimes I_{B_jC_j},\Phi_j^\dagger(X)]\otimes I_R)|\psi'\rangle
		&=\int  \langle \varphi'| \left([P_j\otimes I_{B_jC_j},(\mathcal{V}^{(q,p)}_{\Lambda_j})^\dagger(X)\otimes m(q,p)_{B_jC_j}]\otimes I_R\right) |\psi '\rangle
		dq\,dp\\
		&=\int  \langle \varphi'| (\mathcal{V}^{(q,p)}_{\Lambda_j})^\dagger\left([\mathcal{V}^{(q,p)}_{\Lambda_j}(P_j)\otimes I_{B_jC_j},X\otimes m(q,p)_{B_jC_j}]\otimes I_R\right) |\psi'\rangle
		dq\,dp\ .
	\end{align*}
	Since $\mathcal{V}^{(q,p)}_{\Lambda_j}(P_j)=P_j+p I_{\Lambda}$ for all $(q,p)$, it follows that 
	\begin{align*}
		\Delta_j&=\langle \varphi'| ([P_j\otimes I_{B_jC_j},\Phi_j^\dagger(X)]\otimes I_R)|\psi'\rangle
		=\int  \langle \varphi'| (\mathcal{V}^{(q,p)}_{\Lambda_j})^\dagger\left([P_j\otimes I_{B_jC_j},X\otimes m(q,p)_{B_jC_j}]\otimes I_R\right) |\psi'\rangle
		dq\,dp\\
		&=\int  \langle \varphi'|\left( (\mathcal{V}^{(q,p)}_{\Lambda_j})^\dagger\left([P_j,X]\right)\otimes m(q,p)_{B_jC_j}\otimes I_R\right) |\psi'\rangle
		dq\,dp\\
		&=\langle \varphi'|(\Phi_j^\dagger([P_j,X])\otimes I_R)|\psi'\rangle\ .
	\end{align*}
	Now, by assumption, the channel $\Phi_j$ is moment-limited by \Cref{Steanemoment-limited}, so we can use \Cref{lem.technical} to argue that, for there exist a family $\{\alpha_{v}\}_v$ of positive numbers in $\ell_1$ with $\|v\mapsto \alpha_{v}\|_{\ell_1}\le 1$, and two families $\{|e_{v}\rangle\}_v$ and $\{|f_{v}\rangle\}_v$ of orthonormal vectors in $\operatorname{dom}(\sqrt{N_\Lambda}\otimes I_R)$ such that
	\begin{align}
		\Delta_j=\sum_{v}\, \alpha_{v}\,\langle f_{v}|([P_j,X]\otimes I_R)|e_{v}\rangle \,.\label{eq:Deltaj}
	\end{align}
	We next analyze the expression~$\Delta'_j$. We have 
	\begin{align}
		\left[P_{C_j}\otimes I_{\Lambda B_j},\Phi_j^\dagger(X)\right]\otimes I_R&=\int \left(\left[P_{C_j}\otimes I_{\Lambda B_j}, (\cV^{(q,p)}_{\Lambda_j})^\dagger(X)\otimes m(q,p)_{B_jC_j}\right]\otimes I_R\right)dq\,dp\nonumber\\
		&=\int \left((\cV^{(q,p)}_{\Lambda_j})^\dagger(X)\otimes \left[P_{C_j}\otimes I_{B_j}, m(q,p)_{B_jC_j}\right]\otimes I_R\right)dq\,dp\nonumber\\
		&=\int \left(e^{ipQ_j}e^{-iqP_j}Xe^{iqP_j}e^{-ipQ_j}\otimes m_{\hat{P}}(p)_{B_j}\otimes \left[P_{C_j}, m_{\hat{Q}}(q)_{C_j}\right]\otimes I_R\right)dq\,dp\nonumber\\
		&=\int
		\left(\Psi_j^\dagger(e^{-iqP_j}Xe^{iqP_j})_{\Lambda B_j}\otimes \left[P, m_{\hat{Q}}(q)\right]_{C_j}\otimes I_R\right)dq\ ,\label{eq:Pcjbjexpr}
	\end{align}
	where we introduced the completely positive unital map
	\begin{align*}
		\Psi_j^\dagger :X\mapsto \,\int e^{ipQ_j}Xe^{-ipQ_j}\otimes m_{\hat{P}}(p)_{B_j}\,dp\,.
	\end{align*}
	Once again, $\Psi_j$ is moment-limited by \Cref{Steanemoment-limited}, so we can use \Cref{lem.technical} and Eq.~\eqref{eq:Pcjbjexpr} to argue that there exist a family $\{\beta_{v}\}_v$ of positive numbers and two families $\{|g_{v}\rangle\}_v$ and $\{|h_{v}\rangle\}_v$ of orthogonal vectors in $\operatorname{dom}(\sqrt{N_\Lambda+N_{C_j}}\otimes I_R)$ such that 
	\begin{align*}
		\Delta_j'
		&=\sum_v\beta_{v}\,
		\int_{\mathbb{R}}    \langle h_{v}|\left(e^{-iqP_j}Xe^{iqP_j}\otimes [P,m_{\hat{Q}}(q)]_{C_j}\otimes I_R\right)|g_{v}\rangle\,dq 
	\end{align*}
	with $\|v\mapsto \beta_{v}\|_{\ell_1}\le 1$.
	Lemma~\ref{lem:Pmoveeq} below gives an alternative expression for the integral, which allows us to rewrite this as 
	\begin{align*}
		\Delta_j'
		&=-\sum_v\beta_{v}\,
		\int_{\mathbb{R}}  \langle h_{v}|\,\left(e^{-iqP_j}[P_j,X]e^{iqP_j}\otimes m_{\hat{Q}}(q)_{C_j}\otimes I_R\right)|g_{v}\rangle\,dq\\
		&=-\langle \varphi'|\left(\Phi_j^\dagger\big([P_j,X]\big)\otimes I_R\right)|\psi'\rangle\\
		&=-\Delta_j
	\end{align*}
	where we re-summed over $v$ in the second identity above. Therefore, we have arrived at the following identity
	\begin{align}\label{eq111}
		\langle \varphi|\,\big[P_j,\Phi^{\operatorname{St}\dagger}_j(X)\big]\otimes I_R|\psi\rangle= 2\Delta_j=2\sum_{v}\alpha_{v}\,\langle f_{v}|[P_j,X]\otimes I_R|e_{v}\rangle\,.
	\end{align}

	One can also show a similar expression for $ |\langle\varphi|\left([Q_j,\Phi^{\operatorname{St}\dagger}_j (X)]\otimes I_R\right)|\psi\rangle|$: following the steps leading to \Cref{eq:expansionexprone}, we find that 
	\begin{align}
		&\langle \varphi|\left(\left[P_j,(\Phi^{\operatorname{St}\dagger})^{\otimes m}(X)\right]\otimes I_R\right)\ket{\psi}=\Gamma_j+\Gamma_j'\,,
	\end{align}
	where
	\begin{align}
		&\Gamma_j =\langle \varphi'|\,\big(\big[Q_j\otimes I_{B_jC_j},\Phi_j^\dagger(X)\big]\otimes I_R\big)|\psi'\rangle\\
		&\Gamma_j'=\langle \varphi'|\,\big(\big[Q_{B_j}\otimes I_{\Lambda C_j},\Phi_j^\dagger(X)\big]\otimes I_R\big)|\psi'\rangle\,.\label{Gammajprime}
	\end{align}
	Then, a similar analysis to that performed for $\Delta_j$, using that $\cV^{(q,p)}_{\Lambda_j}(Q_j)=Q_j+qI_\Lambda$, leads to
	\begin{align*}
		\Gamma_j= \langle \varphi'|\,\big(\Phi_j^\dagger([Q_j,X])\otimes I_R\big)|\psi'\rangle\,.
	\end{align*}
	Similarly, the analysis of $\Gamma_j'$ follows that of $\Delta_j'$. We have
	\begin{align}
		\left[Q_{B_j}\otimes I_{\Lambda C_j},\Phi_j^\dagger(X)\right]\otimes I_R&=\int \left(\left[Q_{B_j}\otimes I_{\Lambda C_j}, (\cV^{(q,p)}_{\Lambda_j})^\dagger(X)\otimes m(q,p)_{B_jC_j}\right]\otimes I_R\right)dq\,dp\nonumber\\
		&=\int \left((\cV^{(q,p)}_{\Lambda_j})^\dagger(X)\otimes \left[Q_{B_j}\otimes I_{C_j}, m(q,p)_{B_jC_j}\right]\otimes I_R\right)dq\,dp\nonumber\\
		&=\int \left(e^{ipQ_j}e^{-iqP_j}Xe^{iqP_j}e^{-ipQ_j}\otimes m_{\hat{Q}}(q)_{C_j}\otimes \left[Q_{B_j}, m_{\hat{P}}(p)_{B_j}\right]\otimes I_R\right)dq\,dp\nonumber\\
		&=\int
		\left(\Upsilon_j^\dagger(e^{ipQ_j}Xe^{-ipQ_j})_{\Lambda C_j}\otimes \left[Q, m_{\hat{P}}(p)\right]_{B_j}\otimes I_R\right)dp\ ,
	\end{align}
	where we introduced the completely positive unital map
	\begin{align*}
		\Upsilon_j^\dagger :X\mapsto \,\int e^{-iqP_j}Xe^{iqP_j}\otimes m_{\hat{Q}}(q)_{C_j}\,dq\,.
	\end{align*}
	Once again, $\Upsilon_j$ is moment-limited by \Cref{Steanemoment-limited}, so we can use \Cref{lem.technical} and Eq.~\eqref{eq:Pcjbjexpr} to argue that there exist a family $\{\gamma_{v}\}_v$ of positive numbers and two families $\{|p_{v}\rangle\}_v$ and $\{|q_{v}\rangle\}_v$ of orthogonal vectors in $\operatorname{dom}(\sqrt{N_\Lambda+N_{B_j}}\otimes I_R)$ such that 
	\begin{align*}
		\langle \varphi'|\left(\left[Q_{B_j}\otimes I_{\Lambda C_j},\Phi_j^\dagger(X)\right]\otimes I_R\right)|\psi'\rangle
		&=\sum_v\gamma_{v}\,
		\int_{\mathbb{R}}    \langle q_{v}|\left(e^{ipQ_j}Xe^{-ipQ_j}\otimes [Q,m_{\hat{P}}(p)]_{B_j}\otimes I_R\right)|p_{v}\rangle\,dp
	\end{align*}
	with $\|v\mapsto \gamma_{v}\|_{\ell_1}\le 1$.
	Lemma~\ref{lem:Pmoveeq} below gives an alternative expression for the integral, which allows us to rewrite this as 
	\begin{align*}
		\langle \varphi'|\left(\left[Q_{B_j}\otimes I_{\Lambda C_j},\Phi_j^\dagger(X)\right]\otimes I_R\right)|\psi'\rangle
		&=-\sum_v\gamma_{v}\,
		\int_{\mathbb{R}}  \langle q_{v}|\,\left(e^{ipQ_j}[Q_j,X]e^{-ipQ_j}\otimes m_{\hat{P}}(p)_{B_j}\otimes I_R\right)|p_{v}\rangle\,dq\\
		&=-\langle \varphi'|\left(\Phi_j^\dagger\big([Q_j,X]\big)\otimes I_R\right)|\psi'\rangle\,,
	\end{align*}
	where we re-summed over $v$ in the second identity above. By inserting this expression into~\eqref{Gammajprime},  we obtain 
	\begin{align}
		\Gamma_j'&=-\langle \varphi'|\left(\Phi_j^\dagger\big([Q_j,X]\big)\otimes I_R\right)|\psi'\rangle=-\Gamma_j\ .
	\end{align}
	Therefore, 
	\begin{align}\label{eq:333}
		\langle \varphi|\,\big[Q_j,\Phi^{\operatorname{St}\dagger}_j(X)\big]\otimes I_R|\psi\rangle=0\,.
	\end{align}
	By combining \eqref{eq111} and \eqref{eq:333}, we have  
	\begin{align*}
		\|\nabla_j \Phi^{\operatorname{St}\dagger}_j(X)\|&=\,\max_{R_j\in\{Q_j,P_j\}}\sup_{R}\,\sup_{|\psi\rangle,|\varphi\rangle}\big|\langle\varphi|[R_j,\Phi^{\operatorname{St}\dagger}_j(X)]\otimes I_R|\psi\rangle\big|\\
		&=2\sup_{R}\,\sup_{|\psi\rangle,|\varphi\rangle}\,|\Delta_j|\\
		&=2\,\sup_{|\psi\rangle,|\varphi\rangle}\, \left|\sum_{v}\, \alpha_{v}\,\langle f_{v}|([P_j,X]\otimes I_R)|e_{v}\rangle\right|\\
		&\overset{(1)}{\le} 2\,\sup_{R'}\sup_{|\psi'\rangle, |\varphi'\rangle}\,\,\big|\langle \varphi'|\,[P_j,X]\otimes I_{R'}|\psi'\rangle\big|\\
		&\le 2\,\|\nabla_j X\|\,,
	\end{align*}
	where $(1)$ follows from choosing $R'=RV$, $|\varphi'\rangle=\sum_{v}\alpha_{v}\,|f_{v}\rangle\otimes |v\rangle_{V}$ and $|\psi'\rangle=\sum_{v}\alpha_{v}\,|e_{v}\rangle\otimes |v\rangle_{V}$.
\end{proof}

\begin{lemma}\label{lem:Pmoveeq}
	Let $\ket{e},\ket{f} \in \operatorname{dom}(\sqrt{N_\Lambda+N_{C_j}}\otimes I_R)$. Then 
	\begin{align*}
		&\int_{\mathbb{R}}\langle f|\left(e^{-iqP_j}X e^{iqP_j}\otimes [P,m_{\hat{Q}}(q)]_{C_j}\otimes I_R\right)|e\rangle\,dq \,=-
		\int_{\mathbb{R}} \langle f|\left(\,e^{-iqP_j}[P_j,X]e^{iqP_j}\otimes m_{\hat{Q}}(q)_{C_j}\,\otimes I_R\right)|e \rangle\,dq\  .
	\end{align*}
	Similarly, for any $|e\rangle,|f\rangle\in\operatorname{dom}(\sqrt{N_\Lambda+N_{B_j}}\otimes I_R)$:
	\begin{align*}
		&\int_{\mathbb{R}}    \langle f|\left(e^{ipQ_j}Xe^{-ipQ_j}\otimes [Q,m_{\hat{P}}(p)]_{B_j}\otimes I_R\right)|e\rangle\,dp\,=-\int_{\mathbb{R}}\,\langle f|\,\left(e^{ipQ_j}[Q_j,X]e^{-ipQ_j}\otimes m_{\hat{P}}(p)_{B_j}\otimes I_R\ \right)|e\rangle\,dp
	\end{align*}
	
\end{lemma}
\begin{proof}
	Let us define
	\begin{align}
		G(q)\coloneqq e^{-iqP_j}X e^{iqP_j}\otimes m_{\hat{Q}}(q)_{C_j}\otimes I_R\qquad\textrm{ and }\qquad g(q)\coloneqq \langle f|G(q)|e\rangle\ .
	\end{align}
	We compute the derivative $g'(q)=\langle f|G'(q)|e\rangle$.
	Because $m_{\hat{Q}}(q)=\frac{1}{2\pi\sqrt{\alpha_q}}e^{-iqP}e^{-\frac{1}{2\alpha_q}Q^2}e^{iqP}$ we have 
	\begin{align*}
		m_{\hat{Q}}'(q)&=-[iP,e^{-iqP}m_{\hat{Q}}(0)e^{iqP}]=-i[P,m_{\hat{Q}}(q)]\ 
	\end{align*}
	and thus
	\begin{align}
		G'(q)&=e^{-iqP_j}[-iP_j,X]e^{iqP_j}\otimes m_{\hat{Q}}(q)_{C_j}\otimes I_R-e^{-iqP_j}Xe^{iqP_j}\otimes [iP,m_{\hat{Q}}(q)]\otimes I_R\ .\label{eq:Gderivativem}
	\end{align}
	Eq.~\eqref{eq:Gderivativem} implies that the claim of the Lemma is equivalent to the statement that~$\int_{\mathbb{R}} g'(q)dq=0$. By the fundamental theorem of calculus, it thus suffices to show that
	\begin{align}
		\lim_{|q|\rightarrow\infty }g(q)&=0\ . \label{eq:asymptoticsg}
	\end{align}
	In order to show~\eqref{eq:asymptoticsg}, we introduce the following representations of the bipartite states $|f\rangle$ and $|e\rangle$:
	\begin{align*}
		|f\rangle=\sum_a\,\sqrt{\lambda_a}\,|\psi_a\rangle_{\Lambda R}\otimes |a\rangle_{C_j}\qquad \text{ and }\qquad|e\rangle=\sum_a\,\sqrt{\mu_a}\,|\varphi_a\rangle_{\Lambda R}\otimes |a\rangle_{C_j}
	\end{align*}
	where $\{|a\rangle\}$ is the eigenbasis of $N_{C_j}$. Since both states belong to the domain of $I_{\Lambda R}\otimes \sqrt{N_{C_j}}$, we have that
	\begin{align}
		\langle f|I_{\Lambda R}\otimes N_{C_j}|f\rangle =\sum_a\,\lambda_a \,a<+\infty\qquad \Rightarrow \qquad \lambda_a=\underset{a\to\infty}{o}(a^{-2})\,\label{lambdaasconvergence},
	\end{align}
	and similarly $\mu_a=\underset{a\to\infty}{o}(a^{-2})$. Next, we consider the decomposition $g(q)=\sum_{aa'}\sqrt{\lambda_a\mu_{a'}}\,g_{aa'}(q)$ with overlaps
	\begin{align*}
		g_{aa'}:q\mapsto  \langle \psi_a|\,\big(e^{-iqP_j}Xe^{iqP_j}\big)\,\otimes I_R|\varphi_{a'}\rangle\,\,\langle a|m_{\hat{Q}}(q)_{C_j}|a'\rangle\,.
	\end{align*}
	We analyze the second overlap in the definition of $g_{aa'}$: denoting by $\Pi_Q$ the projection-valued measure of the position operator $Q$, we have
	\begin{align*}
		m_{\hat{Q}}(q)= \frac{1}{2\pi\sqrt{\alpha_q}}e^{-\frac{1}{2\alpha_q}(Q-q)^2}=\frac{1}{2\pi\sqrt{\alpha_q}}\int_{\mathbb{R}}\,e^{-\frac{1}{2\alpha_q}(y-q)^2}\,\Pi_Q(dy)\,.
	\end{align*}
	Therefore, for any $q\in\mathbb{R}$,
	\begin{align}
		|q|\, \langle a|m_{\hat{Q}}(q)|a\rangle&=  \frac{1}{2\pi}\,\int_{\mathbb{R}}\,\frac{1}{\sqrt{\alpha_q}}|q|\,e^{-\frac{1}{2\alpha_q}(y-q)^2}\,\langle a|\Pi_Q(dy)|a\rangle\label{eqmqq}\\
		&\overset{(1)}{\le} \frac{\sqrt{\alpha_q}}{\pi}\,\sup_{x\in\mathbb{R}}\,|x|e^{-x^2}+\frac{1}{2\pi\sqrt{\alpha_q}}\,\int_{\mathbb{R}}\,|y|\langle a|\Pi_Q(dy)|a\rangle\nonumber\\
		&=\frac{\sqrt{\alpha_q}\,e^{-\frac{1}{4}}}{2\pi}+\frac{1}{4\pi^2{\alpha_q}}\langle a||Q|| a\rangle\nonumber\\
		&\le \frac{\sqrt{\alpha_q}\,e^{-\frac{1}{4}}}{2\pi}+\frac{1}{2\pi\sqrt{{\alpha_q}}}\,\sqrt{\langle a|Q^2|a\rangle}\nonumber\\
		&\overset{(2)}{\le}  \frac{\sqrt{\alpha_q}\,e^{-\frac{1}{4}}}{2\pi}+\frac{1}{2\pi\sqrt{\alpha_q}}\,\sqrt{2a+1}\,.\nonumber
	\end{align}
	Above, in $(1)$ we simply used the triangle inequality $|q|\le |q-y|+|y|$ together with the fact that $\langle a|\Pi_Q(dy)|a\rangle$ defines a probability measure; and in $(2)$ we used that $Q^2\le Q^2+P^2=2N+1$. Hence, 
	\begin{align}
		|q|\,|\langle a|m_{\hat{Q}}(q)|a'\rangle|&\le |q|\,\Big(\langle a|m_{\hat{Q}}(q)|a\rangle\,\langle a'|m_{\hat{Q}}(q)|a'\rangle\Big)^{\frac{1}{2}}\nonumber\\
		&\le |q|\,\Big( \frac{\sqrt{\alpha_q}\,e^{-\frac{1}{4}}}{2\pi}+\frac{1}{2\pi\sqrt{\alpha_q}}\,\sqrt{2a+1}\Big)^{\frac{1}{2}}\,\Big( \frac{\sqrt{\alpha_q}\,e^{-\frac{1}{4}}}{2\pi}+\frac{1}{2\pi\sqrt{\alpha_q}}\,\sqrt{2a'+1}\Big)^{\frac{1}{2}}\,.\label{eq2gaa'}
	\end{align}
	Therefore, for any $q\in\mathbb{R}\backslash \{0\}$, we have 
	\begin{align*}
		|g(q)|&\le \sum_{aa'}\,\sqrt{\lambda_a\mu_{a'}}\,|g_{aa'}(q)|\\
		&\le \sum_{aa'}\,\sqrt{\lambda_a\mu_{a'}}\,\|X\|_\infty\,|\langle a|m_{\hat{Q}}(q)_{C_j}|a'\rangle|\\
		&\le \frac{\|X\|_\infty}{|q|}\,\sum_a\,\left(\frac{\lambda_a\sqrt{\alpha_q}\,e^{-\frac{1}{4}}}{2\pi}+\frac{\lambda_a}{2\pi\sqrt{\alpha_q}}\,\sqrt{2a+1}\right)^{\frac{1}{2}}\,\sum_{a'}\,\left(\frac{\mu_{a'}\sqrt{\alpha_q}\,e^{-\frac{1}{4}}}{2\pi}+\frac{\mu_{a'}}{2\pi\sqrt{\alpha_q}}\,\sqrt{2a'+1}\right)^{\frac{1}{2}}\,.
	\end{align*}
	Since $\lambda_a,\mu_a=o(a^{-2})$, the series above converge and therefore $g(q)\to0$ as $|q|\to\infty$, as required.

	The second identity claimed proceeds similarly: we define
	
	\begin{align}
		H(p)\coloneqq e^{ipQ_j}Xe^{-ipQ_j}\otimes m_{\hat{P}}(p)_{B_j}\otimes I_R\qquad\textrm{ and }\qquad h(p)\coloneqq \langle f|h(p)|e\rangle\ .
	\end{align}
	We compute the derivative $h'(p)=\langle f|H'(p)|e\rangle$.
	Because $m_{\hat{P}}(p)=\frac{1}{2\pi\sqrt{\alpha_p}}e^{ipQ}e^{-\frac{1}{2\alpha_p}P^2}e^{-ipQ}$ we have
	\begin{align*}
		m_{\hat{P}}'(p)&=[iQ,e^{ipQ}m_{\hat{P}}(0)e^{-ipQ}]=i[Q,m_{\hat{P}}(p)]\ 
	\end{align*}
	and thus
	\begin{align}
		H'(p)&=e^{ipQ_j}[iQ_j,X]e^{-ipQ_j}\otimes m_{\hat{P}}(p)_{B_j}\otimes I_R+e^{ipQ_j}Xe^{-ipQ_j}\otimes [iQ,m_{\hat{P}}(p)]\otimes I_R\ .\label{eq:Gderivativem}
	\end{align}
	Eq.~\eqref{eq:Gderivativem} implies that the claim of the Lemma is equivalent to the statement that~$\int_{\mathbb{R}} h'(q)dq=0$. Again, it suffices to show that
	\begin{align}
		\lim_{|p|\rightarrow\infty }h(p)&=0\ . \label{eq:asymptoticsg}
	\end{align}
	For this, we consider the decomposition $h(p)=\sum_{aa'}\sqrt{\lambda_a\mu_{a'}}\,h_{aa'}(p)$ with overlaps
	\begin{align*}
		h_{aa'}:p\mapsto  \langle \psi_a|\,\big(e^{ipQ_j}Xe^{-ipQ_j}\big)\,\otimes I_R|\varphi_{a'}\rangle\,\,\langle a|m_{\hat{P}}(p)_{B_j}|a'\rangle\,.
	\end{align*}
	We analyze the second overlap in the definition of $h_{aa'}$: as for $m_{\hat{Q}}(q)$ above, we have
	\begin{align}
		|p|\,|\langle a|m_{\hat{P}}(p)|a'\rangle|
		\le |p|\,\Big( \frac{\sqrt{\alpha_p}\,e^{-\frac{1}{4}}}{2\pi}+\frac{1}{2\pi\sqrt{\alpha_p}}\,\sqrt{2a+1}\Big)^{\frac{1}{2}}\,\Big( \frac{\sqrt{\alpha_p}\,e^{-\frac{1}{4}}}{2\pi}+\frac{1}{2\pi\sqrt{\alpha_p}}\,\sqrt{2a'+1}\Big)^{\frac{1}{2}}\,.\label{eq2haa'}
	\end{align}
	
	Therefore, for any $p\in\mathbb{R}\backslash \{0\}$, we have 
	\begin{align*}
		|h(p)|\le \frac{\|X\|_\infty}{|p|}\,\sum_a\,\left(\frac{\lambda_a\sqrt{\alpha_p}\,e^{-\frac{1}{4}}}{2\pi}+\frac{\lambda_a}{2\pi\sqrt{\alpha_p}}\,\sqrt{2a+1}\right)^{\frac{1}{2}}\,\sum_{a'}\,\left(\frac{\mu_{a'}\sqrt{\alpha_p}\,e^{-\frac{1}{4}}}{2\pi}+\frac{\mu_{a'}}{2\pi\sqrt{\alpha_p}}\,\sqrt{2a'+1}\right)^{\frac{1}{2}}\,.
	\end{align*}
	Since $\lambda_a,\mu_a=o(a^{-2})$, the series above converge and therefore $h(p)\to0$ as $|p|\to\infty$, as required.

\end{proof}

\begin{lemma}\label{Steanemoment-limited}
	For all $j\in[m]$, the maps $\Phi_j$, $\Psi_j$, $\Upsilon_j$ and $\Phi^{\operatorname{St}}_j$ are moment-limited, i.e. for all $\rho\in \cD(\cH_m)$,
	\begin{align*}
		\tr[\rho N]<\infty \quad \Rightarrow\quad \tr[(\Phi_j\otimes \id_{j^c})(\rho)N],\,\tr[(\Psi_j\otimes \id_{j^c})(\rho)N],\,\tr[(\Upsilon_j\otimes \id_{j^c})(\rho)N],\,\tr[(\Phi^{\operatorname{St}}_j\otimes \id_{j^c})(\rho)N]<\infty\,.
	\end{align*}
	
\end{lemma}
\begin{proof}
	The maps $\Psi_j,\Phi_j$ and $\Upsilon_j$ are Gaussian and therefore moment-limited (see \cite{Shirokov2020}). As for $\Phi^{\operatorname{St}}_j$, it is the composition of the channel $\Phi_\sigma:\rho\mapsto \rho\otimes \sigma_{B_jC_j}$, where the approximate $\mathsf{GKP}$ state $\sigma_{B_jC_j}$ has finite energy, with a Gaussian channel. $\Phi^{\operatorname{St}}_j$ is therefore also moment-limited.
\end{proof}

\subsection{Boundedness under concatenated stabilizer decoding}
In this section, we consider recovery maps obtained from code concatenation. Specifically, we consider  an $[m,k,d]$-stabilizer code $\cC$ - where $m$ denotes the number of physical qubits, $k$ the number of logical qubits and $d$ the code distance- with stabilizer generators~$\{S_j\}_{j=1}^{m-k}$  and an associated recovery map~$\cR$ of the form
\begin{align}
	\cR(\rho)&=\sum_{s\in \{0,1\}^{m-k}}C(s)\Pi(s)\rho\Pi(s)C(s)^\dagger\ ,
\end{align}
where for $s\in \{0,1\}^{m-k}$, the operator $\Pi(s)=\prod_{j=1}^{m-k}\frac{1}{2}(I+(-1)^{s_j}S_j)$ projects onto the syndrome-$s$ subspace, and $C(s)$ is a Pauli-correction classically computed from~$s$, which maps the state back to the code space. 

We are interested in local recovery operations, and formalize this as follows: We assume that there is a  Clifford circuit~$U$ (composed of one- and two-qubits) on the system and additional~$m-k$ auxiliary qubits $A_1\dots A_{m-k}$ such that syndrome information can be extracted by applying~$U$ and subsequently measuring each individual qubits: We have 
\begin{align}
	\Pi(s)\rho\Pi(s)&=\tr_{A^{m-k}}\left( \bigotimes_{j=1}^{m-k} |s_j\rangle \langle s_j|_{A_j} \cU(\rho\otimes |0\rangle\langle 0|^{\otimes m-k})\right)\qquad\textrm{ for every }\qquad s\in \{0,1\}^{m-k}\ .
\end{align}
We capture locality as follows: We assume that the backward lightcone of every 
ancilla qubit has size
\begin{align}
	|\cL^{\leftarrow}_{\cU}(A_j)|&\leq \ell_{\mathrm{meas}}\qquad\textrm{ for }\qquad j=1,\ldots,m-k\ .\label{eq:finitelightconeassumption}
\end{align}
Furthermore, we assume that the correction operation can be computed by applying local functions to~$s$: There is a partition $[m]={\bigcup}_{j=1}^r \cF_j$ of $[m]$ into disjoint subsets~$\cF_1,\ldots,\cF_r$ and functions $C_j:\{0,1\}^{\cF_j}\rightarrow\{I,X,Y,Z\}^{\otimes |\cF'_j|}$ such that 
\begin{align}
	C(s)&=\prod_{j=1}^r C_j(s_{\cF_j})_{\cF'_j}\qquad\textrm{ for all }\qquad s\in \{0,1\}^{m-k}\ ,\label{eq:correctionlocalfunction}
\end{align}
where $s_{\cF_j}$ denotes the restriction of~$s$ to~$\cF_j$, i.e., the corresponding substring of syndrome bits.   That is, for each~$j=1,\ldots,r$,  each operator~$C_j(s_{\cF_j})$ has support contained in a set~$\cF'_j$. For what follows, we further assume that the sets $\{\cF'_j\}_j$ are pairwise disjoint.

To formalize locality, assume that there is a constant~$\ell_{\mathrm{corr}}$ such that 
\begin{align}
	|\cF_j| &\leq \ell_{\mathrm{corr}}\qquad\textrm{ for all }\qquad j=1,\ldots,r\ .
\end{align}
and a constant~$\ell'_{\mathrm{corr}}$ such that 
\begin{align}
	|\cF'_j| & \leq \ell'_{\mathrm{corr}}\qquad\textrm{ for all }\qquad j=1,\ldots,r\ .
\end{align}
In the following, let $\overline{\cR}$ denote  a $\mathsf{GKP}$-encoded version of the original recovery map~$\cR$.
Let us specify this map in more detail. It has the form
\begin{align}
	\overline{\cR}(\rho)=\sum_{s}V_{f(s)} \tr_{A^{(m-k)}}\left[\bigotimes_{j=1}^{m-k}\overline{P}(s_j)_{A_j} \overline{\cU}\big(\rho\otimes\,\bigotimes_{j=1}^{m-k}\sigma_{A_j}\big)\right]\,V_{f(s)}^\dagger\ .\label{eq:recoverymapbosonicgkpencoded}
\end{align}
In this expression,  the bosonic ancilla states $\sigma_{A_j}$ are assumed to have finite moments -- typically, they are approximate $\mathsf{GKP}$ states. The unitary~$\overline{\cU}$ is a circuit composed of Gaussian one- and two-mode unitaries, e.g., a product of CNOT gates if the stabilizer code~$\cC$ is of CSS-type. It acts on $m$~system modes augmented by $m-k$~auxiliary modes. The
assumption~\eqref{eq:finitelightconeassumption}
amounts to the following: denoting
the systems's  quadrature  operators by $(R_1,\ldots,R_{2m})=(Q_1,P_1,\ldots,Q_m,P_m)$, and those of the auxiliary modes by~$\{R_{2m+j}\}_{j=1}^{2(m-k)}$, we have the condition
\begin{align}
	\mathcal{U}(R_{2m+j})=\sum_{i=1}^{2m+2(m-k)}u_{j,i}\,R_i\,,\qquad\text { with }\qquad 
	\left|\{i\in \{1,\ldots,2m+2(m-k)\}\ |\ u_{j,i}\neq 0\}\right|&\leq \ell_{\textrm{meas}}\ .\label{eq:ellmeasdef}
\end{align}
for every $j=1,\ldots,2(m-k)$.

The projection $\overline{P}(s_k)_{A_k}$ corresponds to a measurement of logical Paulis on the ancilla $A_k$. In practice, this can be done by an approximate homodyne measurement in the form of Equations \eqref{homodyneapprox} and \eqref{homodyneapproxbis} modulo a fixed real parameter. Since for approximate $\mathsf{GKP}$ states, the measurement generally gives real outcomes, we choose $s_k\in\{0,1\}$ by rounding the outcome. More precisely, each POVM element $\overline{P}(s_i)$ can be written as an integral over the operator density $m_{\hat{Q}}$ or $m_{\hat{P}}$, for instance
\begin{align*}
	\overline{P}(0)=\int_{\Gamma} m_{\hat{P}}(z)\,dz\,,\qquad \overline{P}(1)=I-\overline{P}(0)=\int_{\Gamma^c}\,m_{\hat{P}}(z)\,dz\,,
\end{align*}
where $\Gamma$ denotes the domain of the real line such that whenever the parameter $z\in\Gamma$ is measured, it is interpreted as syndrome $0$. 

Finally, the correction is done via a displacement~$V_{f(s)}$ of parameter $f(s)\in\mathbb{R}^{2m}$:
This is the $\mathsf{GKP}$-encoded counterpart of applying the Pauli correction~$C(s)$.
Recalling that logical $\overline{X}$, $\overline{Y}$- and $\overline{Z}$-operators in the $\mathsf{GKP}$ code are given by the displacements in \Cref{GKPlogicalpauli}, namely
\begin{align}
	\overline{X}=e^{-i\sqrt{\pi}\hat{P}},\qquad \overline{Z}=e^{i\sqrt{\pi}\hat{Q}},\qquad \overline{Y}=e^{i\sqrt{\pi}(\hat{Q}-\hat{P})}\,,
\end{align}
we  conclude that
\begin{align}
	\|f(s)\|_\infty &\leq \sqrt{\pi}\qquad\textrm{ for every }\qquad s\in \{0,1\}^{m-k}\ .\label{eq:boundednessoff}
\end{align}
It follows from Eq.~\eqref{eq:correctionlocalfunction} that the function~$f:\{0,1\}^{m-k}\rightarrow \mathbb{R}^{2m}$ satisfies the following: We have
\begin{align}
	f(s)&=\left(\bigoplus_{u=1}^r f_u(s_{\cF_u})_{\cF'_u}\right)\oplus \left(0^{2(m-\sum_{u=1}^r |\cF'_u|)}\right)_{[m]\backslash \bigcup_{u=1}^r \cF'_u}\qquad\textrm{ for every }\qquad s\in \{0,1\}^{m-k}\ .\label{eq:fsdirectsum}
\end{align}
for some functions $f_u:\{0,1\}^{\cF_u}\rightarrow \{0,1\}^{\cF'_u}$.
In particular, $V_{f(s)}=\prod_{u=1}^r V_{f_u(s_{\cF_u})_{\cF'_u}}$ is a 
product of (possibly multi-mode) displacements on disjoint subsets~$\{\cF'_u\}_u$ of modes, each depending on at most~$\ell_{\mathrm{corr}}$ syndrome bits.
This completes the specification of the recovery map~$\overline{\cR}$ in~\eqref{eq:recoverymapbosonicgkpencoded}. Our main result for this recovery map is the following:
\begin{prop}\label{prop:upperboundlocalrecovery}
	The recovery map $\overline{\cR}$ satisfies  
	\begin{align*}
		\|\overline{\cR}\|_{ W_{\operatorname{B}}\to W_{\operatorname{B}}}\le  \left(\ell_{\operatorname{meas}}+ \frac{\ell_{\operatorname{meas}}\,2^{1+\ell_{\operatorname{meas}}(1+\ell_{\operatorname{corr}})}\ell_{\operatorname{meas}}\cdot \ell_{\operatorname{corr}}'}{\sqrt{\alpha_{\min}\pi}}\right)\,.
	\end{align*}
\end{prop}
Before discussing the proof of~\Cref{prop:upperboundlocalrecovery}, let us briefly summarize some of the relevant properties of the function~$f$ derived from Eqs.~\eqref{eq:boundednessoff} and~\eqref{eq:fsdirectsum}. For any two syndromes $s,s'\in \{0,1\}^{m-k}$ differing on $\delta$ fixed ancillas, $f(s)-f(s')$ only depends on at most $\delta\cdot \ell_{\mathrm{corr}}$ ancilla bits,
and has support on at most $\delta\cdot \ell'_{\mathrm{corr}}$ modes.  
More precisely, defining 
\begin{align}
	D\coloneqq \{j\in [m-k]\ |\ s_j\neq s'_j\}
\end{align}
as the set of indices where the syndromes~$s$ and $s'$ differ, there are $v\in \{0,1\}^{|D|}$ $w\in \{0,1\}^{m-k-|D|}$ such that 
\begin{align}
	\begin{matrix}
		s_{D}&=&v\\
		s'_{D}&=&\overline{v}
	\end{matrix}\qquad \begin{matrix}
		s_{D^c}&=&w\\
		s'_{D^c}&=&w
	\end{matrix}
\end{align}
where $D^c=[m-k]\backslash D$ and $\overline{v}_j=1-v_j$ for each $j\in \big[|D|\big]$. It follows immediately that for $u\in [r]$, we have 
\begin{align}
	s_{\cF_u}\neq s'_{\cF_u}\qquad\textrm{ if and only if }\qquad D\cap \cF_u\ne \emptyset\ .
\end{align}
Omitting ``padding'' zeros, we thus have (using Eq.~\eqref{eq:fsdirectsum})
\begin{align}
	f(s)-f(s')&=\bigoplus_{u|D\cap \cF_u\ne \emptyset}\left(f_u(s_{\cF_u})-f_u(s'_{
		\cF_u})\right)_{\cF_u'}=\bigoplus_{u|D\cap\cF_u\ne \emptyset}g_u\left(v,s_{\cF_u\backslash D}\right)_{\cF_u'}\label{eq:ssprimfm}
\end{align}
with $g_u(v,t)\coloneqq f_u(v_{D\cap \cF_u}\oplus t_{\cF_u\backslash D})-f_u(\overline{v}_{D\cap \cF_u}\oplus t_{\cF_u\backslash D})$.  In this expression, the only dependence on~$w=s_{D^c}$ is through the collection of substrings~$\{s_{\cF_u\backslash D}\}_{u\in D}$, hence the total number of bits of~$w$ the expression~$f(s)-f(s')$ depends on is bounded by
\begin{align}
	|\{u|D\cap \cF_u\ne \emptyset\}|\cdot  \max_{u|D\cap\cF_u\ne\emptyset}|\cF_u\backslash D|\le|D|\cdot\max_{u\in D}|\cF_u\backslash D|\leq |D|\cdot \ell_{\mathrm{corr}}\,.
\end{align}
In summary, there is a subset~$M\subset D^c$ of size $|M|\leq |D|\cdot\ell_{\mathrm{corr}}$ and a function~$g$ such that 
\begin{align}
	f(s)-f(s')&=g(v,w_{M})\qquad\textrm{ for all }v,w\qquad\textrm{ where } s=s(v,w),s'=s'(v,w)\qquad\textrm{ is given by~\eqref{eq:ssprimfm}}\ .\label{eq:localitydifferencem}
\end{align}
Observe  that~\eqref{eq:ssprimfm} also implies that the number of non-zero components of $g(v,w_M)$ is bounded by 
\begin{align}
	\left|\left\{j\in [2m]\ |\ g_j(v,w_M)\neq 0\right\}\right| \le 2|\{u|D\cap \cF_u\ne \emptyset\}|\cdot \max_{u|D\cap\cF_u\ne\emptyset}|\cF'_u|  \leq 2|D|\cdot \ell'_{\mathrm{corr}}\ ,\label{eq:upperboundfsfsprime}
\end{align}
that is, $f(s)-f(s')$ has support on at most $|D|\cdot\ell'_{\mathrm{corr}}$ modes, as claimed.

\begin{proof}[Proof of \Cref{prop:upperboundlocalrecovery}]
	Once again, we argue by duality: we first fix a reference system $R$, choose a vertex $j\in [m]$ and consider a self-adjoint, bounded operator $X$ as well as $|\psi\rangle,|\varphi\rangle\in\operatorname{dom}(\sqrt{N}\otimes I_R)$. Then with a slight abuse of notation, we denote $\overline{P}(s)\coloneqq \bigotimes_{j=1}^{m-k}\overline{P}(s_j)$ for a given syndrome $s=(s_1,\cdots,s_{m-k})$, $X(s)\coloneqq V_{f(s)}^\dagger X V_{f(s)}$, $\sigma_{A^{(m-k)}}=\bigotimes_{j=1}^{m-k}\sigma_{Aj}$, and compute
	
	\begin{align*}
		\overline{\cR}^{\dagger}(X)&=\sum_{s\in \{0,1\}^{m-k}}\, \tr_{A^{(m-k)}}\Big[(I\otimes \sigma_{A^{(m-k)}}) \,\overline{\mathcal{U}}^\dagger\,\big(X(s) \otimes \overline{P}(s)\big)\Big]\,.
	\end{align*}
	Next, we consider a quadrature $R_{2j+1}=P_j$ corresponding to a vertex $j\in[m]$ and partition each syndrome $s=(v^{(j)},w^{(j)})$ onto a part $v^{(j)}$ of measurement outcomes $s_k$ corresponding to projections $\overline{P}(s_k)$ of support overlapping with the light-cone of vertex $j$ with respect to the unitary map $\overline{\mathcal{U}}$, and a part $w^{(j)}$ of measurement outcomes corresponding to projections whose support does not overlap with that light-cone. That is, we set 
	\begin{align}
		v^{(j)}=\{s_r\}_{r\textrm{: }j\in \cL^{\leftarrow}_{\cU}(A_r)}\qquad\textrm{ and }\qquad 
		w^{(j)}=\{s_r\}_{r\textrm{: }j\not\in \cL^{\leftarrow}_{\cU}(A_r)}
	\end{align}
	Let us reorder (relabel) the syndrome bits in such a way that the first $\delta_j$ bits  constitute~$v^{(j)}$ and the remaining~$m-k-\delta_j$ bits constitute~$w^{(j)}$.  By assumption, we have $\delta_j\leq \ell_{\operatorname{meas}}$. Then we have 
	\begin{align}
		\sum_{s\in \{0,1\}^{m-k}} X(s)\otimes \overline{P}(s)&=
		\sum_{(v,w)\in \{0,1\}^{\delta_j}\times \{0,1\}^{m-k-\delta_j}} X(v,w)\otimes \overline{P}(v,w)\nonumber\\
		&=\sum_{w\in \{0,1\}^{m-k-\delta_j}} X(0^{\delta_j},w)\otimes P(0^{\delta_j},w)+
		\sum_{\substack{v\in \{0,1\}^{\delta_j}\backslash\{0^{\delta_j}\}\\
				w\in \{0,1\}^{m-k-\delta_j}}} X(v,w)\otimes \overline{P}(v,w)\nonumber\\
		&=\sum_{w\in \{0,1\}^{m-k-\delta_j}}X(0^{\delta_j},w)\otimes \overline{P}(w)+
		\sum_{\substack{v\in \{0,1\}^{\delta_j}\backslash\{0^{\delta_j}\}\\
				w\in \{0,1\}^{m-k-\delta_j}}}\Delta_vX(w)\otimes \overline{P}(v,w)\label{eq:commutativityfirsterm}
	\end{align}
	where we introduced the operators 
	\begin{align*}
		\overline{P}(w)&\coloneqq \sum_{v\in \{0,1\}^{\delta_j}}\overline{P}(v,w)\\
		\Delta_v X(w)&\coloneqq X(v,w)-X(0^{\delta_j},w)
	\end{align*}
	for $w\in \{0,1\}^{m-k-\delta_j}$. 
	By definition, $\overline{P}(w)$ acts non-trivially only on auxiliary qubits~$A_i$ 
	with the property that $j\not\in \cL^{\leftarrow}_{\cU}(A_i)$. Note that we have 
	\begin{align}
		\overline{\cU}(P_j)&=\sum_{i=1}^{2m} u_{2j+1,i}R_i+\sum_{i\in \{1,\ldots,m-k\}: j\in\cL_{\cU}^\leftarrow(A_i)}^{m-k} \left(u_{2j+1,2m+2i-1}R_{2m+2i-1}+u_{2j+1,2m+2i}R_{2m+2i}\right)
	\end{align}
	since all terms $i$ with $j$ not belonging to the backward lightcone of~$A_i$ do not contribute, i.e.,
	\begin{align}
		u_{2j+1,2m+2i-1}=u_{2j+1,2m+2i}=0\qquad\textrm{ for all }j\not\in \cL_{\cU}^\leftarrow(A_i)\ .\label{eq:zerolightconeprop}
	\end{align}
	It thus follows with Eq.~\eqref{eq:commutativityfirsterm}, \eqref{eq:commutativityfirsterm} and the support property of~$\overline{P}(w)$ that
	\begin{align}
		[\,\overline{\cU}(P_j),\Lambda]
		&=\sum_{i=1}^{2m}u_{2j+1,i}  [R_i,\Lambda]+\sum_{i=1}^{2(m-k)} u_{2j+1,2m+i}\,[R_{2m+i},\Lambda']
	\end{align}
	where we used the abbreviations
	\begin{align}
		\Lambda\coloneqq \sum_{s\in \{0,1\}^{m-k}}X(s)\otimes \overline{P}(s)\qquad\textrm{ and }\qquad \Lambda'\coloneqq \sum_{\substack{v\in \{0,1\}^{\delta_j}\backslash\{0^{\delta_j}\}\\
				w\in \{0,1\}^{m-k-\delta_j}}}\Delta_vX(w)\otimes \overline{P}(v,w)\label{eq:expressionsimplified}\,.
	\end{align} 
	With~\eqref{eq:expressionsimplified}, we obtain
	\begin{align}
		[P_j,\overline{\cR}^{\dagger}(X)]&=\sum_{s\in \{0,1\}^{m-k}}\,\tr_{A^{(m-k)}}\Big[(I\otimes \sigma_{A^{(m-k)}})[P_j,\overline{\mathcal{U}}^\dagger(X(s)\otimes \overline{P}(s))]\Big]\nonumber\\
		&=\sum_{s\in \{0,1\}^{m-k}}\,\tr_{A^{(m-k)}}\Big[(I\otimes \sigma_{A^{(m-k)}})\,\overline{\mathcal{U}}^\dagger\,\big([\,\overline{\mathcal{U}}(P_j),X(s)\otimes \overline{P}(s)]\big)\Big]\nonumber\\
		&=\Delta_j+\Delta_j'\label{eq:deltaplusdeltaprimexpression}
	\end{align}
	where 
	\begin{align}
		\Delta_j &\coloneqq \sum_{i=1}^{2m}u_{2j+1,i}\sum_{s\in \{0,1\}^{m-k}}\,\tr_{A^{(m-k)}}\Big[(I\otimes\sigma_{A^{(m-k)}})\,\overline{\mathcal{U}}^\dagger \big(\big[R_i,X(s) \big]\otimes \overline{P}(s)\big)\Big]\\
		\Delta_j'&\coloneqq \sum_{i=1}^{2(m-k)}u_{2j+1,2m+i}
		\sum_{\substack{v\in \{0,1\}^{\delta_j}\backslash\{0^{\delta_j}\}\\
				w\in \{0,1\}^{m-k-\delta_j}}}\tr_{A^{(m-k)}}\Big[(I\otimes\sigma_{A^{(m-k)}})\,\overline{\mathcal{U}}^\dagger\big(\Delta_v X(w)\otimes [R_{2m+i},\overline{P}(v,w)]\big)\Big]\,.
	\end{align}
	Next, we bound each of these sums separately. We start with $\Delta_j$: first, we observe by the commutation relations between quadratures and displacements that $[R_i,X(s)]=V_{f(s)}^\dagger[R_i,X]V_{f(s)}$. Therefore, we can rewrite $\Delta_j$ as
	\begin{align*}
		\Delta_j=\sum_{i=1}^{2m}u_{2j+1,i}\,\Phi^\dagger([R_i,X])\quad\text{ with }\quad \Phi^\dagger(Y)\coloneqq \sum_{s\in \{0,1\}^{m-k}}\tr_{A^{(m-k)}}\Big[(I\otimes\sigma_{A^{(m-k)}})\,\overline{\mathcal{U}}^\dagger \big(V_{f(s))}^\dagger YV_{f(s)}\otimes \overline{P}(s)\big)\Big]\,.
	\end{align*}
	Then we use \Cref{lem.technical} together with \Cref{energylimitedstab} to argue that since $|\varphi\rangle,|\psi\rangle\in \operatorname{dom}(\sqrt{N}\otimes I_R)$, there exists a sequence $\{\alpha_u\}$ of positive numbers with $\|u\mapsto \alpha_u\|_{\ell_1}\le 1$, and sequences $\{|e_u\rangle\}$ and $\{|f_u\rangle\}$ of orthogonal vectors in $\operatorname{dom}(\sqrt{N}\otimes I_R)$ such that
	\begin{align}\label{Z1}
		\langle \varphi|(\Delta_j\otimes I_R)|\psi\rangle =\sum_{i=1}^{2m}\,u_{2j+1,i}\,\sum_{u}\,\alpha_u\,\langle e_u|\,[R_i,X]\,\otimes I_R|f_u\rangle\,.
	\end{align}
	Therefore,
	\begin{align}
		| \langle \varphi|(\Delta_j\otimes I_R)|\psi\rangle |\le \sum_{i=1}^{2m}\,|u_{2j+1,i}|\,\|\nabla X\|\le \ell_{\operatorname{meas}}\,\|\nabla X\|\,.\label{Deltajeq}
	\end{align}
	where in the second inequality we used that the number of indices $i\in [2(m-k)]$ such that $u_{2j+1,2m+1}\neq 0$ is upper bounded by~$\ell_{\textrm{meas}}$, see Eq.~\eqref{eq:ellmeasdef}.
	Next, we consider $\langle \varphi|(\Delta_j'\otimes I_R)|\psi\rangle$: recalling that the state $\sigma_{A^{(m-k)}}$ is a tensor product of pure $\mathsf{GKP}$ states which we denote by $|\widetilde{\psi}_{A^{(m-k)}}\rangle$, we can rewrite the quantity as 
	\begin{align}
		\langle \varphi|(\Delta_j'\otimes I_R)|\psi\rangle&=\sum_{i=1}^{2(m-k)}u_{2j+1,2m+i}
		\sum_{\substack{v\in \{0,1\}^{\delta_j}\backslash\{0^{\delta_j}\}\\
				w\in \{0,1\}^{m-k-\delta_j}}}\, \langle \widetilde{\varphi}| \left(\,\overline{\mathcal{U}}^\dagger\big(\Delta_v X(w)\otimes [R_{2m+i},\overline{P}(v,w)]\big)\otimes I_R\right)|\widetilde{\psi}\rangle\label{eq:deltapidphipsi}
	\end{align}
	where we denoted $|\widetilde{\varphi}\rangle\coloneqq |\varphi\rangle\otimes |\widetilde{\psi}_{A^{(m-k)}}\rangle$, $|\widetilde{\psi}\rangle\coloneqq |\psi\rangle\otimes |\widetilde{\psi}_{A^{(m-k)}}\rangle$.
	By definition, we have with the definition $\cV_x(\cdot)=V_{x}^\dagger(\cdot)V_x$ for $x\in \mathbb{R}^{2m}$ that 
	\begin{align}
		\Delta_v X(w)&= \cV_{f(v,w)}(X)-\cV_{f(0^\delta_j,w)}(X)\nonumber\\
		&=\cV_{f(v,w)}\left(X-\cV_{f(v,w)}^\dagger(\cV_{f(0^\delta_j,w)}(X))\right)\nonumber\\
		&=\cV_{f(v,w)}\left(X-\cV_{-\Delta_vf(w)}(X)\right)\label{eq:translationidentitydeltavf}
	\end{align}
	where we used that $\cV_\alpha\circ\cV_\beta=\cV_{\alpha+\beta}$ and $\cV_\alpha^\dagger=\cV_{-\alpha}$ and the abbreviation 
	\begin{align}
		\Delta_vf(w)\coloneqq f(v,w)-f(0^{\delta_j},w)\qquad\textrm{ for }\qquad v\in \{0,1\}^{\delta_j}\textrm{ and }w\in \{0,1\}^{m-k-\delta_j}\ .\label{eq:deltavfwabbrev}
	\end{align}
	From the local structure of the function~$f$, (i.e., expression~\eqref{eq:fsdirectsum} and the fact that $|\cF_j|\leq \ell_{\mathrm{corr}}$), it follows that we can decompose the vector $w$ into two subvectors~$w=(w_1,w_2)$ where $w_1\in \{0,1\}^{\delta_j\cdot \ell_{\mathrm{corr}}}$,  such that
	\begin{align}
		\Delta_vf(w)&=\delta_vf(w_1)
	\end{align}
	for a function~$\delta_vf:\{0,1\}^{\delta_j\cdot\ell_{\mathrm{corr}}}\rightarrow \mathbb{R}^{2m}$, see Eq.~\eqref{eq:localitydifferencem}.
	

	
	For fixed $v\in \{0,1\}^{\delta_j}\backslash \{0^{\delta_j}\}$
	and $w=(w_1,w_2)\in \{0,1\}^{m-k-\delta_j}$ we have by Eq.~\eqref{eq:translationidentitydeltavf}
	\begin{align*}
		\Delta_v X(w)\otimes [R_{2m+i},\overline{P}(v,w)]
		&= \big(\cV_{f(v,w)}(X-\cV_{-\Delta_v f(w)}(X))\big)\otimes [R_{2m+i},\overline{P}(v,w)]\\
		&= \big(\cV_{f(v,w)}(X-\cV_{-\delta_v f(w_1)}(X))\big)\otimes [R_{2m+i},\overline{P}(v,w_1)\overline{P}(w_2)]\\
		&= \big(\cV_{f(v,w)}(
		\delta_v(X)(w_1)
		)\big)\otimes [R_{2m+i},\overline{P}(v,w_1)\overline{P}(w_2)]\,,
	\end{align*}
	where
	\begin{align}
		\delta_v(X)(w_1)&\coloneqq X-\cV_{-\delta_v f(w_1)}(X)\label{eq:deltavxdef}\,.
	\end{align}
	It follows that
	\begin{align}
		&\sum_{\substack{w=(w_1,w_2)\\
				\in \{0,1\}^{m-k-\delta_j}}}
		\langle \widetilde{\varphi}| \left(\,\mathcal{U}^\dagger\big(\Delta_v X(w)\otimes [R_{2m+i},\overline{P}(v,w)]\big)\otimes I_R\right)|\widetilde{\psi}\rangle\label{eq:sumtildephipsisimple}\\
		&\qquad\qquad \qquad \qquad\qquad\qquad =\,\sum_{w_1\in \{0,1\}^{\delta_j\cdot \ell_{\mathrm{corr}}}} \langle \widetilde{\varphi} |\,\left(\Psi_{w_1}^\dagger 
		\left((\delta_v(X)(w_1)\otimes [R_{2m+i},\overline{P}(v,w_1)]\right)\otimes I_R\right)|\widetilde{\psi}\rangle\nonumber\,,
	\end{align} 
	where
	\begin{align}
		\Psi^\dagger_{w_1}(Z)&\coloneqq \sum_{w_2\in \{0,1\}^{m-k-\delta_j(\ell_{\mathrm{corr}}+1)}}\overline{\cU}^\dagger\left(
		\cV_{f(v,(w_1,w_2))}(Z)\otimes \overline{P}(w_2)\right)\ .
	\end{align}
	The map~$\Psi^\dagger_{w_1}$ is completely positive, unital and moment-limited by \Cref{energylimitedstab}.  Therefore, by \Cref{lem.technical} together with \Cref{energylimitedstab} there exists a sequence $\{\alpha_u^{(w_1)}\}_u$ of positive numbers with $\|u\mapsto \alpha_u^{(w_1)}\|_{\ell_1}\le 1$, and sequences $\{|e_u^{(w_1)}\rangle\}_u$ and $\{|f_u^{(w_1)}\rangle\}_u$ of vectors in $\operatorname{dom}(\sqrt{N}\otimes I_R)$ such that 
	\begin{align*}
		\langle \widetilde{\varphi} |\,\left(\Psi_{w_1}^\dagger 
		\left((\delta_v(X)(w_1)\otimes [R_{2m+i},\overline{P}(v,w_1)]\right)\otimes I_R\right)|\widetilde{\psi}\rangle
		&=\sum_{u}\alpha_u^{(w_1)}\langle e_u^{(w_1)}|(\delta_v X(w_1)\otimes [R_{2m+i},\overline{P}(v,w_1)]\otimes I_R)|f_u^{(w_1)}\rangle\ .
	\end{align*}
	We show in Lemma~\ref{lem:upperboundcommutatornorm} below that
	\begin{align}
		\|[R_{2m+i},\overline{P}(v,w_1)]\|&\leq \frac{1}{\pi\sqrt{\alpha_q}}\ \label{eq:toprovertwomi}
	\end{align}
	when $i$ is even (i.e., $R_{2m+i}$ is a momentum-operator). 
	Defining the normalized vectors
	\begin{align}
		|\tilde{f}^{(w_1)}_u\rangle \coloneqq \frac{( [R_{2m+i},\overline{P}(v,w_1)]\otimes I_R)|f^{(w_1)}_u\rangle}{\| ([R_{2m+i},\overline{P}(v,w_1)]\otimes I_R)|f^{(w_1)}_u\rangle\|}\ ,
	\end{align}
	it thus follows that
	\begin{align*}
		\left|\langle \widetilde{\varphi} |\,\left(\Psi_{w_1}^\dagger 
		\left((\delta_v(X)(w_1)\otimes [R_{2m+i},\overline{P}(v,w_1)]\right)\otimes I_R\right)|\widetilde{\psi}\rangle\right|
		&\leq \frac{1}{\pi\sqrt{\alpha_q}}\sum_{u}\alpha_u^{(w_1)}\left|\langle e_u^{(w_1)}|(\delta_v X(w_1)\otimes I_{A^{(m-k)}R})|\tilde{f}_u^{(w_1)}\rangle\right|\ .
	\end{align*}
	Applying~\eqref{eq:upperboundsimplifiedinnerprod} from Lemma~\ref{lem:translateupper} to each term in this sum gives
	\begin{align}
		\left|\langle \widetilde{\varphi} |\,\left(\Psi_{w_1}^\dagger 
		\left((\delta_v(X)(w_1)\otimes [R_{2m+i},\overline{P}(v,w_1)]\right)\otimes I_R\right)|\widetilde{\psi}\rangle\right|
		&\leq \frac{2\ell_{\operatorname{meas}}\cdot\ell_{\mathrm{corr}}'}{\sqrt{\alpha_q\pi }}\|\nabla X\|\ .\label{eq:upperboundnablaxs}
	\end{align}
	Inserting~\eqref{eq:upperboundnablaxs}  into~\eqref{eq:sumtildephipsisimple} 
	and combining this with~\eqref{eq:deltapidphipsi} gives
	\begin{align}
		\left|\langle \varphi|(\Delta_j'\otimes I_R)|\psi\rangle\right|&\leq \sum_{i=1}^{2(m-k)}|u_{2j+1,2m+i}|
		\sum_{\substack{v\in \{0,1\}^{\delta_j}\backslash\{0^{\delta_j}\}\\
				w_1\in \{0,1\}^{\delta_j\cdot \ell_{\mathrm{corr}}}}}\,\left|\langle \widetilde{\varphi} |\,\left(\Psi_{w_1}^\dagger 
		\left((\delta_v(X)(w_1)\otimes [R_{2m+i},\overline{P}(v,w_1)]\right)\otimes I_R\right)|\widetilde{\psi}\rangle\right|
		\\
		&\leq \frac{2^{1+\ell_{\operatorname{meas}}(1+\ell_{\operatorname{corr}})}\ell_{\operatorname{meas}}\cdot \ell_{\operatorname{corr}}'}{\sqrt{\alpha_q\pi}}\,\|\nabla X\|\cdot \sum_{i=1}^{2(m-k)}|u_{2j+1,2m+i}|\ .
	\end{align}
	Recall that the number of indices $i\in [2(m-k)]$ such that $u_{2j+1,2m+1}\neq 0$ is upper bounded by~$\ell_{\textrm{meas}}$, see Eq.~\eqref{eq:ellmeasdef} hence
	\begin{align}
		\left|\langle \varphi|(\Delta_j'\otimes I_R)|\psi\rangle\right|&\leq \frac{\ell_{\operatorname{meas}}\,2^{1+\ell_{\operatorname{meas}}(1+\ell_{\operatorname{corr}})}\ell_{\operatorname{meas}}\cdot \ell_{\operatorname{corr}}'}{\sqrt{\alpha_q\pi}}\,\|\nabla X\|\ .\label{eq:deltaprimeupperboundtwo}
	\end{align}
	Combining~\eqref{Deltajeq} and ~\eqref{eq:deltaprimeupperboundtwo} with Eq.~\eqref{eq:deltaplusdeltaprimexpression} gives
	\begin{align}
		\left|\langle \varphi| [P_j,\overline{\cR}^\dagger(X)]\psi\rangle\right|\le \left(\ell_{\operatorname{meas}}+ \frac{\ell_{\operatorname{meas}}\,2^{1+\ell_{\operatorname{meas}}(1+\ell_{\operatorname{corr}})}\ell_{\operatorname{meas}}\cdot \ell_{\operatorname{corr}}'}{\sqrt{\alpha_q\pi}}\right)\,\|\nabla X\|\,.
	\end{align}
	Repeating the steps that led to this bound for $P_j$, with $P_j$ replaced by $Q_j$, we then have found that
	\begin{align*}
		\|\nabla \Phi^{\operatorname{Stab}\dagger}(X)\|\le  \,\left(\ell_{\operatorname{meas}}+ \frac{\ell_{\operatorname{meas}}\,2^{1+\ell_{\operatorname{meas}}(1+\ell_{\operatorname{corr}})}\ell_{\operatorname{meas}}\cdot \ell_{\operatorname{corr}}'}{\sqrt{\alpha_{\operatorname{min}}\pi}}\right) \|\nabla X\|\ .
	\end{align*}
	The claim then follows by duality, cf.~\eqref{eq:duality}.
\end{proof}

\begin{lemma}\label{lem:upperboundcommutatornorm}
	We have 
	\begin{align}
		\|[P_{m+i},\overline{P}(v,w_1)]\|&\leq \frac{1}{\pi\sqrt{\alpha_q}}\ .
	\end{align}
\end{lemma}
\begin{proof}
	$R_i=P$ is a local momentum operator, so that it is enough to consider that $P(s_i)=\int_{\Gamma} m_{\hat{P}}(z)\,dz$
	where $\Gamma$ denotes the domain of the real line such that whenever the parameter $z\in\Gamma$ is measured, it is interpreted as syndrome $s_i$. Then,
	\begin{align*}
		|[R_i,P(s_{i})]|&=\Big|\int_\Gamma\,[R_i,m_{\hat{P}}(z)]\,dz\Big|\\
		&=\frac{1}{2\pi\sqrt{\alpha_q}}\,\Big|\int_\Gamma\,V_{(0,-z)^t}[P,e^{-\frac{1}{2\alpha_q}Q^2}]V_{(0,z)^t}\,dz\Big|\\
		&=\frac{1}{2\pi{\alpha_q}^{\frac{3}{2}}}\,\Big|\int_{\Gamma}\,V_{(0,-z)^t}\,Qe^{-\frac{1}{2\alpha_q}Q^2}\,V_{(0,z)^t}\,dz\Big|\\
		&= \frac{1}{2\pi{\alpha_q}^{\frac{3}{2}}}\,\Big|\int_{\Gamma}\,(Q-z)e^{-\frac{1}{2\alpha_q}(Q-z)^2}\,dz\Big|\\
		&     \le \frac{1}{2\pi{\alpha_q}^{\frac{3}{2}}}\,\int\,|Q-z|\,e^{-\frac{1}{2\alpha_q}(Q-z)^2}\,dz\\
		&=\frac{1}{\pi\sqrt{\alpha_q}}\,I\,,
	\end{align*}
	and therefore $\|[R_i,P(s_i)]\|_\infty\le (\pi\sqrt{\alpha_q})^{-1}$.
	
\end{proof}

\begin{lemma}\label{lem:translateupper}
	Let $\ket{e}$ and~$|f\rangle $ be vectors in 
	$\operatorname{dom}(\sqrt{N_\Lambda+N_{C_j}}\otimes I_R)$. Then  there is a distribution over pairs~$(|e'\rangle,|f'\rangle)$ of vectors in   $\operatorname{dom}(\sqrt{N_\Lambda+N_{C_j}}\otimes I_R)$ such that 
	\begin{align*}
		\left|\langle e|(\delta_v X(w_1)\otimes I_{A^{(m-k)}R})|f\rangle\right| &\leq 
		2\sqrt{\pi}\ell_{\operatorname{meas}}\cdot \ell'_{\mathrm{corr}}\cdot\max_{i\in [2m]}\mathbb{E}_{|e'\rangle,|f'\rangle}\Big[
		\left|\langle e'| ([R_{i},X]\otimes I_{A^{(m-k)}R})|f'\rangle\right|
		\Big]\ .
	\end{align*}
	In particular, we have 
	\begin{align}
		\left|\langle e|(\delta_v X(w_1)\otimes I_{A^{(m-k)}R})|f\rangle\right| &\leq 
		2\sqrt{\pi}\,\ell_{\operatorname{meas}}\cdot \ell'_{\mathrm{corr}}\cdot \|\nabla X\|\ .\label{eq:upperboundsimplifiedinnerprod}
	\end{align}
\end{lemma}
\begin{proof}
	Recall (cf. Eq.~\eqref{eq:deltavxdef}) that 
	$\delta_v(X)(w_1)\coloneqq X-\cD_{-\delta_v f(w_1)}(X)$
	is the difference of two displaced operators. Define $\xi\coloneqq -\delta_v f(w_1)$. Recall that by
	Eq.~\eqref{eq:upperboundfsfsprime}, the number of non-zero entries of~$\xi$ is upper bounded by~$ \ell'_{\textrm{corr}}\cdot  \ell_{\operatorname{meas}}$ and $\|\xi\|_\infty\leq 2\sqrt{\pi}$ by \eqref{eq:boundednessoff}. 
	We now consider  a linear interpolation between $0$ and $\xi\in \mathbb{R}^{2m}$: We have 
	\begin{align*}
		|   \langle e| \delta_v X(w_1)\otimes I_{A^{(m-k)}R}|f\rangle|&=
		\Big|\int_0^1\,\frac{d}{dx}\,\langle e|(V_{x\xi}^\dagger\,X\,V_{x\xi}\otimes I_{A^{(m-k)}R})|f\rangle \,dx\Big|\nonumber \\
		&\leq \int_0^1 \sum_{i=1}^{2m} |\xi_i|\cdot
		\,\left|\langle e|(V_{x \xi}^\dagger [R_i,X]V_{x\xi}\otimes I_{A^{(m-k)}R})|f\rangle\right| \,dx\Big|\\
		&\leq 2\sqrt{\pi}\,\ell_{\operatorname{meas}}\cdot \ell'_{\textrm{corr}} \cdot \max_{i\in[2m]}\,\int_0^1 
		\left|\langle V_{x\xi}e| ([R_i,X]\otimes I_{A^{(m-k)}R})|V_{x\xi}f\rangle\right| \,dx\,.
	\end{align*}
	This implies the claim.

\end{proof}

\begin{lemma}\label{energylimitedstab}
	
	The maps $\Phi$ and $\Psi^\dagger_{w_1}$ are moment-limited.
\end{lemma}

\begin{proof}
	This is clear since these maps are linear combinations of compositions of Gaussian channels and partial measurements. The latter are of the form $\rho_{AB}\mapsto \tr_B(I_A\otimes P_B \rho_{AB})$ for some local projection $P_B$, so that 
	\begin{align*}
		\tr(\tr_B(I_A\otimes P_B\rho_{AB})N_A)=\tr(\rho_{AB}N_A\otimes P_B)\le \tr(\rho_{AB}N_A)\le \tr(\rho_{AB}N_{AB})\,.
	\end{align*}
	Therefore, these maps are moment-limited.

\end{proof}

\end{document}